\documentclass[journal,onecolumn]{IEEEtran}
\usepackage[T1]{fontenc}
\usepackage{amsmath,amsfonts,amsthm, mathabx}
\usepackage{algorithmic,lipsum}
\usepackage{array}
\usepackage[caption=false,font=normalsize,labelfont=sf,textfont=sf]{subfig}
\usepackage{textcomp}
\usepackage{stfloats}
\usepackage{float}
\usepackage{hyperref}       
\hypersetup{
  colorlinks   = true, 
  urlcolor     = blue, 
  linkcolor    = blue, 
  citecolor   = blue 
}
\usepackage{url}
\usepackage{float}
\usepackage{verbatim}
\usepackage{graphicx}
\def\BibTeX{{\rm B\kern-.05em{\sc i\kern-.025em b}\kern-.08em
    T\kern-.1667em\lower.7ex\hbox{E}\kern-.125emX}}
\usepackage{bbm}
\usepackage{physics}
\newtheorem{theorem}{Theorem}[section]
\newtheorem{corollary}[theorem]{Corollary}
\newtheorem{lemma}[theorem]{Lemma}
\newtheorem{remark}[theorem]{Remark}
\newtheorem{definition}[theorem]{Definition}
\newtheorem{prop}[theorem]{Proposition}
\newtheorem{example}[theorem]{Example}
\newtheorem{conj}[theorem]{Conjecture}
\usepackage{tikz-cd}

\newcommand{\mb}{\mathbb}
\newcommand{\mc}{\mathcal}
\newcommand{\wt}{\widetilde}

\newcommand{\pw}[1]{{\color{red}[PW: #1]}}

\ifCLASSINFOpdf
\else
\fi
\newcommand{\added}[1]{#1}

\begin{document}




\title{Additivity of quantum capacities in simple non-degradable quantum channels}

\author{Graeme Smith and Peixue Wu
\thanks{GS (graeme.smith@uwaterloo.ca) and PW (p33wu@uwaterloo.ca) are with the Department of Applied Mathematics, and Institute for Quantum Computing at
University of Waterloo, Ontario, Canada. This work was supported by the Canada First Research Excellence Fund (CFREF).
}
}

\maketitle


\begin{abstract}
Quantum channel capacities give the fundamental performance limits for information flow over a communication channel. However, the prevalence of superadditivity is a major obstacle to understanding capacities, both quantitatively and conceptually. In contrast, examples exhibiting additivity, though relatively rare, offer crucial insights into the origins of nonadditivity and form the basis of our strongest upper bounds on capacity. Degradable channels, whose coherent information is provably additive, stand out as among the few classes of channels for which the quantum capacity is exactly computable.

In this paper, we introduce two families of non-degradable channels whose coherent information remains additive, making their quantum capacities tractable. First, we demonstrate that channels capable of “outperforming” their environment, under conditions weaker than degradability, can exhibit either strong or weak additivity of coherent information. Second, we explore a complementary construction that modifies a channel to preserve coherent information additivity while destroying the “outperforming” property. We analyze how structural constraints guarantee strong and weak additivity and investigate how relaxing these constraints leads to the failure of strong additivity, with weak additivity potentially persisting.
\end{abstract}
\tableofcontents

%
\IEEEpeerreviewmaketitle

\ifCLASSOPTIONcaptionsoff
  \newpage
\fi

\section{Introduction}
\IEEEPARstart{A}{}central problem in quantum information theory is to determine the capacities of various quantum channels. If Alice can encode $nR$ units of information using $n$ copies of a quantum channel $\mc N$ with vanishing error as $n\to \infty$, then $R$ is said to be an achievable rate if we send information through $\mc N$. The maximum achievable rate for quantum information(units of qubits), private information(units of bits hidden from the environment) and classical information(units of bits) are defined to be the channel's quantum, private and classical capacity, denoted by $\mc Q,\mc P, \mc C$ respectively. Notably, it was shown in \cite{lloyd1997capacity, shor2002, devetak2005private, barnum1998information}
(LSD Theorem) that the quantum capacity of a quantum channel $\mc N$ is characterized by its coherent information
\begin{equation}\label{multiletter}
    \mc Q(\mc N) = \lim_{n \to \infty} \frac{1}{n}\mc Q^{(1)}(\mc N^{\otimes n}),
\end{equation}
where $\mc Q^{(1)}(\mc N):= \max_{\rho_A} I_c(\rho_A, \mc N)$ is the maximal coherent information. Similarly, this regularization procedure is also required for private capacity \cite{devetak2005private} and classical capacity \cite{schumacher1997sending, holevo1998capacity} (also known as HSW Theorem). By optimizing over product states, one always has super-additivity for any two quantum channels:
\begin{equation}\label{superadditive}
    \mc Q^{(1)}(\mc N \otimes \mc M) \ge \mc Q^{(1)}(\mc N) + \mc Q^{(1)}(\mc M).
\end{equation}
If coherent information were additive, we could remove the regularization process and compute the quantum capacity of a channel as easily as the capacity of a classical noisy channel. However, coherent information is generally not additive. The first explicit demonstration of this fact was provided by DiVincenzo and Shor \cite{divincenzo1998quantum}, who showed that for a certain depolarizing channel $\mathcal{N}$, the single-letter coherent information is strictly less than the coherent information of multiple channel uses, i.e., $\mathcal{Q}^{(1)}(\mathcal{N}) < \frac{1}{n}\mathcal{Q}^{(1)}(\mathcal{N}^{\otimes n})$ for some $n \in \mathbb{N}$. Moreover, a seminal result in \cite{smith2008quantum} showed that there exist two channels $\mc N,\mc M$, each of which has zero quantum capacity on its own, yet $\mc N\otimes \mc M$ has positive capacity. This phenomenon, called \textit{super-activation}, demonstrates an extreme form of non-additivity and shows that the structure of quantum channels can be very subtle, especially when entangled inputs are allowed. Following these examples, numerous studies \cite{cubitt2015unbounded, elkouss2015superadditivity, filippov2021capacity, koudia2022deep, leditzky2023generic, leditzky2023platypus, leditzky2018dephrasure, lim2018activation, lim2019activation, siddhu2020leaking, siddhu2021entropic, siddhu2021positivity, sidhardh2022exploring, smith2008structured, smith2011quantum} have provided further examples of non-additivity. Such non-additive behaviors pose the primary obstacles to the accurate evaluation of quantum capacity.

Although it is well-established that the coherent information of a general quantum channel is not additive, there is still no complete characterization of the circumstances under which additivity holds. In particular, even identifying the full range of channels with zero quantum capacity remains an open problem. Currently, the only two known large classes of zero-capacity channels are the \emph{anti-degradable} channels \cite{devetak2005capacity} and the \emph{PPT (entanglement-binding)} channels \cite{peres1996separability, horodecki2001separability, horodecki2000binding}. It is closely related to a notoriously hard problem to find bound entangled states which are not PPT \cite{Horodecki_2005}. 

For channels with strictly positive single-letter coherent information, only a few families are known to exhibit additivity: (1) the \emph{degradable} channels themselves \cite{devetak2005capacity}, (2) channels that are \emph{dominated by degradable} channels \cite{chessa2021partially, chessa2023resonant, Gao_2018}, and (3) channels satisfying a \emph{weaker notion of degradability} \cite{watanabe2012private} (which includes, as a special case, the \emph{conjugate degradable} channels of \cite{bradler2010conjugate}). Beyond these classes, additivity remains elusive and can fail in dramatic ways.

In this paper, we systematically investigate the conditions under which the coherent information of quantum channels is additive. We distinguish between two notions of additivity: we say that a quantum channel $\mathcal{N}$ has \textit{weakly additive} coherent information if $\mc Q^{(1)}(\mathcal{N}^{\otimes n}) = n\mc Q^{(1)}(\mathcal{N})$ for every $n \in \mathbb{N}$. On the other hand, if equality in the superadditivity condition \eqref{superadditive} holds for any quantum channel $\mathcal{M}$, we say that $\mathcal{N}$ has \textit{strong additive} coherent information. Currently, the only known class of channels with strong additive coherent information is \textit{Hadamard channels}, i.e., the complementary channel is entanglement-breaking \cite{Winter_2016}. 

If the equality in \eqref{superadditive} only holds when $\mathcal{M}$ is restricted to a certain subclass of channels, then we say that $\mathcal{N}$ has strong additive coherent information with respect to that subclass. Previous studies \cite{smith2008quantum, leditzky2023generic, leditzky2023platypus} have shown that some quantum channels may exhibit weak additivity of coherent information, potentially with positive coherent information, while strong additivity can fail even for relatively simple degradable channels, such as erasure channels. These findings suggest that the complete characterization of additivity depends on which notion of additivity one considers. 

The main contribution of this paper is to identify two classes of non-degradable and non-PPT channels that nonetheless exhibit additivity properties. First, we show that if a quantum channel’s output system “outperforms” the environment under conditions weaker than degradability, both \emph{strong} and \emph{weak} additivity of coherent information can be achieved, thereby generalizing earlier work~\cite{watanabe2012private, cross2017uniform}. To illustrate this, we construct a class of examples based on probabilistic mixtures of degradable and anti-degradable channels, which themselves are neither strictly degradable nor anti-degradable, yet still retain the intuitive feature that the output “outperforms” the environment. Under a plausible stability conjecture, these channels can realize both strong and weak additivity. The technical aspects of this conjecture motivate the introduction of a novel framework, termed the \textit{reverse-type data processing inequality}, which is explored in detail in an independent work~\cite{BGSW24}.

Second, we prove that starting with a channel exhibiting strong or weak additivity of coherent information, one can construct a new channel that preserves additivity \emph{but} no longer maintains the “output outperforms the environment” property. Our construction is inspired by the notion of switched channels~\cite{cubitt2015unbounded, elkouss2015superadditivity} and further developed in~\cite{chessa2021partially}. We demonstrate this approach by recovering previous additivity results from~\cite{chessa2021quantum, chessa2023resonant, Gao_2018}, and by generalizing the Platypus channel introduced in~\cite{leditzky2023generic} to show additivity over certain parameter regions. Moreover, we show that when the structural rigidity responsible for additivity is broken, \emph{strong} additivity of coherent information fails, while \emph{weak} additivity is possible to persist. This distinction between strong and weak additivity highlights a rich theoretical landscape for investigating additivity phenomena in quantum information theory. It will not only help us determine the capacities of more quantum channels, but also teach us about when non-additivity can arise.

The rest of the paper is structured as follows. Section \ref{sec:preliminary} reviews the preliminaries on quantum channels and capacity. Section \ref{sec:method} introduces the general construction of quantum channels with strong or weak additivity properties. Sections \ref{sec:flagged channel example} and \ref{sec:Platypus} analyze specific examples in detail, illustrating strong or weak additivity for non-degradable quantum channels.

\section{Preliminaries}\label{sec:preliminary}
\subsection{Quantum channel and its representation}
In this paper, $\mc H$ is denoted as a Hilbert space of finite dimension. $\mc H^{\dagger}$ is the dual space of $\mc H$. $\ket{\psi}$ denotes a unit vector in $\mc H$ and $\bra{\psi} \in \mc H^{\dagger}$ is the dual vector. For two Hilbert spaces $\mc H_A, \mc H_B$, the space of linear operators mapping from $\mc H_A$ to $\mc H_B$ is denoted as $\mb B(\mc H_A, \mc H_B) \cong \mc H_B \otimes \mc H_A^{\dagger}$. When $\mc H_A = \mc H_B = \mc H$, we denote $\mb B(\mc H, \mc H)$ as $\mb B(\mc H)$.

Let $\mc H_A,\mc H_B,\mc H_E$ be three Hilbert spaces of dimensions $d_A,d_B,d_E$ respectively. An isometry $V: \mc H_A \to \mc H_B \otimes \mc H_E$, meaning $V^{\dagger}V = I_{A}$ (identity operator on $\mc H_A$), generates a pair of quantum channels $(\mc N, \mc N^c)$, i.e., a pair of completely positive and trace-preserving(CPTP) linear maps on $\mb B(\mc H_A)$, defined by 
\begin{equation}
    \mc N(\rho) = \Tr_{E}(V \rho V^{\dag}), \quad \mc N^c(\rho) = \Tr_{B}(V \rho V^{\dag}),
\end{equation}
which take any operator $\rho \in \mb B(\mc H_A)$ to $\mb B(\mc H_B)$ and $\mb B(\mc H_E)$, respectively. Each channel in the pair $(\mc N, \mc N^c)$ is called the \textit{complementary channel} of the other. 

Denote $\mc L(\mb B(\mc H_A), \mb B(\mc H_B))$ as the class of super-operators which consists of linear maps taking any operator in $\mb B(\mc H_A)$ to $\mb B(\mc H_B)$. For any $\mc N \in \mc L(\mb B(\mc H_A), \mb B(\mc H_B))$, we denote it as $\mc N^{A\to B}$ to emphasize that it is a super-operator mapping operators on $\mc H_A$ to operators on $\mc H_B$.

In general, any $\mc N^{A\to B}$ has the operator-sum representation
\begin{align*}
    & \mc N^{A\to B}(X) = \sum_{i=1}^m A_i X B_i, \quad A_i \in \mb B(\mc H_A, \mc H_B),\quad B_i \in \mb B(\mc H_B, \mc H_A), \quad X \in \mb B(\mc H_A).
\end{align*}
A quantum channel $\mc N^{A\to B}$ is the one with completely positive and trace-preserving(CPTP) property. The operator-sum representation of a quantum channel is given by $B_i = A_i^{\dagger}$, and in this case, we call it \textit{Kraus representation}:
\begin{equation}
    \mc N^{A\to B}(X) = \sum_{i=1}^m A_i X A_i^{\dagger},\quad A_i \in \mb B(\mc H_A, \mc H_B),\ X \in \mb B(\mc H_A).
\end{equation}
Another representation of a super-operator $\mc N^{A\to B}$ comes from its Choi–Jamio\l{}kowski operator. Suppose $\{\ket{i}\}_{i=0}^{d_A -1}$ is an orthonormal basis for $\mc H_A$ and the maximally entangled state on $\mc H_A \otimes \mc H_A$ is given by 
\begin{align*}
    \ket{\Phi} = \frac{1}{\sqrt{d_A}} \sum_{i=0}^{d_A -1} \ket{i} \otimes \ket{i}.
\end{align*}
Then the \textit{unnormalized Choi–Jamio\l{}kowski} operator of $\mc N^{A\to B}$ is an operator in $\mb B(\mc H_A \otimes \mc H_B)$ given by
\begin{equation}\label{eqn:Choi operator}
\begin{aligned}
     \mc J_{\mc N^{A\to B}} & = d_A (id_{\mb B(\mc H_A)} \otimes \mc N^{A\to B}) (\ket{\Phi} \bra{\Phi}) = \sum_{i,j = 0}^{d_A - 1} \ketbra{i}{j} \otimes \mc N^{A\to B}(\ketbra{i}{j}).
\end{aligned}
\end{equation}
Note that it is well-known that $\mc N$ is completely positive if and only if $\mc J_{\mc N}$ is positive and $\mc N$ is trace-preserving if and only if $\Tr_B(\mc J_{\mc N}) = I_A$, the identity operator on $\mc H_A$, where $\Tr_B$ is the partial trace operator given by $\Tr_B(X_A \otimes X_B) = \Tr(X_B) X_A$. 

The composition rule of Choi–Jamio\l{}kowski operator is given by the well-known link product: suppose $\mc N_1: \mb B(\mc H_A) \to \mb B(\mc H_B),\ \mc N_2: \mb B(\mc H_B) \to \mb B(\mc H_C)$, then 
\begin{equation}
    \mc J_{\mc N_2 \circ \mc N_1} = \Tr_B \left[(I_A \otimes \mc J_{\mc N_2}) (\mc J_{\mc N_1}^{T_B} \otimes I_C)\right],
\end{equation}
where $T_B$ denotes the partial transpose in the $\mc H_B$ system. 

Finally, we review another representation of a super-operator which behaves better under composition. Suppose the operator-sum representation of a super-operator is given by $\mc N(X) = \sum_{k=1}^m A_k X B_k$, we define its \textit{transfer matrix} as an operator in $ \mb B(\mc H_A \otimes \mc H_A, \mc H_B \otimes \mc H_B)$ by 
\begin{equation}\label{def:transfer matrix}
    \mc T_{\mc N} = \sum_{k=1}^m B_k^{T} \otimes A_k.
\end{equation}
It is easy to see that for linear maps $\mc N_1: \mb B(\mc H_A) \to \mb B(\mc H_B),\ \mc N_2: \mb B(\mc H_B) \to \mb B(\mc H_C)$, we have
\begin{equation}\label{eqn:composition of transfer matrix}
    \mc T_{\mc N_2 \circ \mc N_1} = \mc T_{\mc N_2} \mc T_{\mc N_1}.
\end{equation}
Moreover, the connection between Choi–Jamio\l{}kowski opertor and transfer matrix is established as follows:
\begin{equation}\label{eqn:relation Choi and transfer}
    \mc T_{\mc N} = \vartheta^{\Gamma} (\mc J_{\mc N}),
\end{equation}
where $\vartheta^{\Gamma}: \mb B(\mc H_A \otimes \mc H_B) \to \mb B(\mc H_A \otimes \mc H_A, \mc H_B \otimes \mc H_B)$ is the \textit{involution operation} defined by 
\begin{equation}
    \vartheta^{\Gamma} (\ket{j}_A \ket{v}_B \bra{i}_A \bra{r}_B) = \ket{r}_B\ket{v}_B \bra{i}_A\bra{j}_A. 
\end{equation}

\subsection{Quantum capacity and its (non)addivity property}
Suppose a complementary pair of quantum channels $(\mc N, \mc N^c)$ is generated by the isometry $V_{\mc N}: \mc H_A \to \mc H_B \otimes \mc H_E$. The \textit{quantum capacity} of $\mc N$, denoted as $\mc Q(\mc N)$, is the supremum of all achievable rates for quantum information transmission through ${\mc N}$. The LSD theorem \cite{lloyd1997capacity, shor2002, devetak2005private} shows that the coherent information is an achievable rate for quantum communication over a quantum channel. Now we review the basics of coherent information of a quantum channel. 

For any input state $\rho_A \in \mb B(\mc H_A)$, we denote $\ket{\psi}_{RA} \in \mc H_{R} \otimes \mc H_A$ as a purification of $\rho_A$, and $\ket{\Psi}_{RBE} = (I_{R} \otimes V_{\mc N})(\ket{\psi}_{RA})$, where $I_R$ is the identity operator acting on $\mc H_R$ and $V_{\mc N}$ is the isometry generating the quantum channel $\mc N$. We denote $\rho_B,\ \rho_{RB},\ \rho_{E}$ as the reduced density operator from the pure state $\ket{\Psi}_{RBE}$.

The coherent information $I_c(\rho_A,\mc N)$ is defined by
\begin{equation}\label{def:coherent information}
\begin{aligned}
    I_c(\rho_A, \mc N)& := I(R\rangle B)_{\rho_{RB}} = S(\rho_B) - S(\rho_{RB}) = S(\rho_B) - S(\rho_E).
\end{aligned}
\end{equation}
where $S(\rho):= -\Tr(\rho \log \rho)$ is the von Neumann entropy. Note that different choices of the purification system $\mc H_R$ and $\ket{\psi}_{RA}$ will produce the same coherent information, since von Neumann entropy is invariant under unitary transformation. We denote $S(\rho_B)$ as $S(B)$ for simplicity of notation in the remaining of the paper. The maximal coherent information is defined by 
\begin{equation}\label{def:max coherent information}
    \mc Q^{(1)}(\mc N)= \max_{\rho_A} I_c(\rho_A, \mc N)
\end{equation}
and by LSD Theorem, the quantum capacity can be calculated by the regularized quantity $\mc Q(\mc N) = \lim_{n \to \infty} \frac{1}{n}\mc Q^{(1)}(\mc N^{\otimes n})$. In general, the channel coherent information is \textit{super-additive}, i.e., for any two quantum channels $\mc N_1, \mc N_2$, we have
\begin{equation}
    \mc Q^{(1)}(\mc N_1 \otimes \mc N_2) \geq \mc Q^{(1)}(\mc N_1) + \mc Q^{(1)}(\mc N_2).
\end{equation}
We will use the following terminology introduced in \cite{leditzky2023generic} and references therein to facilitate our discussion. 
\begin{itemize}
    \item We say that the quantum channel $\mc N$ has \textit{weak additive} coherent information, if $\mc Q(\mc N) =\mc Q^{(1)}(\mc N)$. 
    \item We say that the quantum channel $\mc N$ has \textit{strong additive} coherent information with a certain class of quantum channels(for example, degradable channels), if for any quantum channel $\mc M$ from that class, we have $\mc Q^{(1)}(\mc N \otimes \mc M) = \mc Q^{(1)}(\mc N) + \mc Q^{(1)}(\mc M)$.
\end{itemize}
The choice of the class of quantum channels matters. Recall that given a complementary pair of channels $(\mc N,\mc N^c)$, we say $\mc N$ is \textit{degradable} and $\mc N^c$ is \textit{anti-degradable}, if there exists a quantum channel $\mc D$ such that $\mc D \circ \mc N= \mc N^c$. We call $\mc D$ a \textit{degrading} channel. We say $\mc N$ is \textit{symmetric}, if $\mc N$ is simultaneous degradable and anti-degradable.

It is well-known that (see \cite[Theorem 13.5.1]{wilde2013quantum}) the class of (anti-)degradable channels have weak additive coherent information and strong additive coherent information with (anti-)degradable channels. Moreover, it is shown that the class of Hadamard channels (its complementary is entanglement-breaking thus anti-degradable) has strong additive coherent information with arbitrary quantum channels \cite{Winter_2016}. Apart from (anti-)degradable channels, PPT channels, i.e., its Choi–Jamio\l{}kowski operator is positive under partial transpose \cite{horodecki2000binding}, have zero quantum capacity via the well-known transpose upper bound of quantum capacity $\mc Q(\mc N)\le \log \|T^B \circ \mc N\|_{\diamond}$, where $T^B$ is the transpose map. Thus the class of PPT channels also has weak additive coherent information. 

According to the knowledge of the authors, there is no rigorous proof of the existence of a quantum channel $\mc N$, with $\mc Q(\mc N) = \mc Q^{(1)}(\mc N)>0$ and a (anti-)degradable channel $\mc M$, such that $\mc Q(\mc N \otimes \mc M)> \mc Q(\mc N) + \mc Q(\mc M)$. In \cite{leditzky2023platypus}, Platypus channel is a candidate of this phenomenon, but the weak additivity is only conjectured to hold.

\subsection{Data processing inequality}
Let $\rho$ be a quantum state and $\sigma \geq 0$. The \textit{relative entropy} is defined as 
    \begin{align}\label{eq:relative-entropy}
        D(\rho \| \sigma) &= \begin{cases}
        \Tr(\rho \log{\rho} - \rho \log{\sigma}) \quad &\text{if } \text{supp}\rho \subseteq \text{supp}\sigma, \\
        \infty \quad &\text{else}.
        \end{cases}
    \end{align}
The well-known \textit{Data-processing inequality}, see \cite{uhlmann1977relative} for original proof and extensions \cite{ohya2004quantum, muller2017monotonicity}, claims that if $\mathcal{N}$ is a positive trace-preserving map(in particular quantum channel), we have
\begin{align}
    D(\rho \| \sigma) \geq D(\mathcal{N}(\rho) \| \mathcal{N}(\sigma)).
\end{align}
Rewriting the mutual information and coherent information in terms of relative entropy, data processing inequality implies the following \cite[Section 11.9]{wilde2013quantum}:
\begin{enumerate}
        \item Suppose $\rho_{VA}$ is a state on $\mc H_V\otimes \mc H_A$ and $\mc N: \mb B(\mc H_A) \to \mb B(\mc H_B)$ is a quantum channel. Denote $\rho_{VB} = (id_{\mb B(\mc H_V)} \otimes \mc N)(\rho_{VA})$ then we have
        \begin{align}\label{DPI:mutual information}
            I(V;A) \ge I(V;B),
        \end{align}
        where the mutual information is defined as 
        \begin{align*}
            I(V;A) &= S(\rho_V) + S(\rho_A) - S(\rho_{VA}) = D(\rho_{VA} \| \rho_V \otimes \rho_A).
        \end{align*}
        \item (Bottleneck inequality) Suppose $\mc N_1: \mb B(\mc H_A) \to \mb B(\mc H_B)$ and $\mc N_2: \mb B(\mc H_B) \to \mb B(\mc H_C)$ are quantum channels, then for any state $\rho_A$, we have 
        \begin{align*}
            I_c(\rho_A, \mc N_2 \circ \mc N_1) \le \min\left\{I_c(\rho_A, \mc N_1), I_c(\mc N_1(\rho_A), \mc N_2) \right\}.
        \end{align*}
        In particular, we have 
        \begin{equation}\label{bottleneck inequality}
        \begin{aligned}
            & \mc Q^{(1)}(\mc N_2 \circ \mc N_1) \le \min\left\{ \mc Q^{(1)}(\mc N_2), \mc Q^{(1)}(\mc N_1)\right\},\\
            & \mc Q(\mc N_2 \circ \mc N_1) \le \min\left\{ \mc Q(\mc N_2), \mc Q(\mc N_1)\right\}.
        \end{aligned}
        \end{equation}
\end{enumerate}

\section{General construction of non-degradable channels with additive coherent information}\label{sec:method}
In this section, we introduce two distinct classes of non-degradable channels that exhibit either weak or strong additivity of coherent information. The first class relies on weaker notions of degradability, while the second is constructed using channels with intrinsic additivity properties, such that although the weaker degradability is disrupted, the additivity property is preserved.

We begin by discussing some foundational constructions, including flagged channels and direct sum channels.
\subsection{Direct sum of channels, flagged channels and their coherent information}
\subsubsection*{Direct sum of channels} 
Suppose we have a finite collection of quantum channels:
\[
  \Phi^{A_k \to B_k} : \mb B(\mc H_{A_k}) \;\longrightarrow\;  \mb B(\mc H_{B_k})
  \quad \text{for } k = 1,2,\dots,n.
\]
The \emph{direct sum} of these channels, denoted
\[
  \bigoplus_{k=1}^n \Phi^{A_k \to B_k} : 
  \mb B\Bigl(\bigoplus_{k=1}^n \mc H_{A_k}\Bigr) 
  \;\longrightarrow\; 
  \mb B\Bigl(\bigoplus_{k=1}^n \mc H_{B_k}\Bigr),
\]
acts on block matrices in a block-diagonal manner, see \cite{Fukuda_2007} for more discussions and \cite{chessa2021partially} for a generalization. Specifically, each 
$\mc H_{A_k}$ and $\mc H_{B_k}$ are treated as orthogonal subspaces, 
and off-diagonal blocks of an operator \(X\) are mapped to zero, while each 
diagonal block \(X_{kk}\) is mapped according to \(\Phi_k\). In other words,
\begin{equation}\label{def:direct sum channel}
    \begin{aligned}
     & \Bigl(\,\bigoplus_{k=1}^n \Phi^{A_k \to B_k} \Bigr)(X) 
  = \begin{pmatrix}
    \Phi^{A_1 \to B_1} (X_{11}) & 0 & \cdots & 0 \\
    0 & \Phi^{A_2 \to B_2}(X_{22}) & \cdots & 0 \\
    \vdots & \vdots & \ddots & \vdots \\
    0 & 0 & \cdots & \Phi^{A_n \to B_n}(X_{nn})
  \end{pmatrix}.
\end{aligned}
\end{equation}

\subsubsection*{Flagged channel}
A quantum channel $\mc N^{A \to FB}$ is called a \emph{flagged channel} if it includes a flag register \(F\) of dimension \(d_F\) and is defined as:
\begin{equation}\label{def:flagged_channel}
    \mc N^{A \to FB} = \sum_{i=0}^{d_F-1} p_i \ketbra{i}{i}^F \otimes \mc N_i^{A \to B}, 
\end{equation}
where \(\sum_i p_i = 1\) and \(p_i \geq 0\). 

An illustrative example is the \emph{erasure channel}:
\[
\mc E_p^{A \to FB} = (1-p)\ketbra{0}{0}^F \otimes id^{A \to B} + p \ketbra{1}{1}^F \otimes \mc E_1^{A \to B},
\]
where $id^{A \to B}$ is the identity channel and \(\mc E_1^{A \to B}\) is a complete erasure channel(also known as \textit{replacer channel}) that maps all input states to a fixed state.

\subsubsection*{Coherent information of direct sum channels and flagged channels}
\begin{lemma}\label{lemma:basic}
For any direct sum channel \(\mc N = \mc N_0 \oplus \mc N_1\) and \(n \geq 1\),
\begin{equation}\label{coh:switch}
    \mc Q^{(1)}(\mc N^{\otimes n}) = \max_{0 \leq \ell \leq n} \mc Q^{(1)}(\mc N_0^{\otimes \ell} \otimes \mc N_1^{\otimes (n-\ell)}).
\end{equation}
For any flagged channel \(\mc N^{A \to FB} = \sum_{i=0}^{d_F-1} p_i \ketbra{i}{i}^F \otimes \mc N_i^{A \to B}\),
\begin{equation}\label{coh:flag}
    I_c(\rho_A, \mc N^{A \to FB}) = \sum_i p_i I_c(\rho_A, \mc N_i^{A \to B}).
\end{equation}
\end{lemma}

\begin{proof}
1. Direct sum channels: Consider the \(n\)-fold tensor product \((\mc N_0 \oplus \mc N_1)^{\otimes n}\). For a \(n\)-bit string \(\mathbf{b} = (b_1, \dots, b_n) \in \{0,1\}^n\), the channel is decomposed as:
   \[
   \mc N^{\otimes n} = \bigoplus_{\mathbf{b} \in \{0,1\}^n} \bigotimes_{i=1}^n \mc N_{b_i}.
   \]
   Using \cite[Proposition 1]{Fukuda_2007}, the coherent information of a direct sum channel is the maximum of its components, therefore
   \begin{align*}
        \mc Q^{(1)}(\mc N^{\otimes n}) & = \max_{\mathbf{b} \in \{0,1\}^n} \mc Q^{(1)}(\bigotimes_{i=1}^n \mc N_{b_i})\\
        & =\max_{0 \leq \ell \leq n} \mc Q^{(1)}(\mc N_0^{\otimes \ell} \otimes \mc N_1^{\otimes (n-\ell)}),
   \end{align*}
where the last equality follows from the fact that the order of tensor products does not affect maximal coherent information.

2. Flagged channels: The orthogonality of states in the flag register \(F\) ensures that the coherent information decomposes as a convex combination of \(I_c\) of the individual channels:
   \[
   I_c(\rho_A, \mc N^{A \to FB}) = \sum_i p_i I_c(\rho_A, \mc N_i^{A \to B}),
   \]
   which can be verified by the definition of coherent information, and von Neumann entropy is additive under convex combinations of orthogonal states.
\end{proof}

\subsection{Additivity via weaker degradability}
Given a complementary pair of channels $(\mc N,\mc N^c)$ generated by isometry $U_{\mc N}: \mc H_A \to \mc H_B \otimes \mc H_E$, it is natural to ask which one is "better" than the other. For a degradable channel $\mc N$, it is better than its complementary channel in the sense that there is another quantum channel $\mc D$ such that $\mc D \circ \mc N = \mc N^c$. Beyond this, various weaker notions of comparison are useful, especially in the study of additivity problems. We briefly review some of these weaker notions, which are systematically studied in \cite{hirche2022contraction}. This includes specific cases covered in \cite{buscemi2012comparison, buscemi2016degradable, watanabe2012private, cross2017uniform}. Before introducing the formal definitions, we establish the notation. Denote $\mc H_V,\mc H_W$ as arbitrary finite dimensional quantum system and $\rho_{VWA}$ is a tripartite quantum state supported on $\mc H_V\otimes \mc H_W \otimes \mc H_A$. Applying the isometry, we obtain the quadripartite state 
\begin{equation}
    \rho_{VWBE} = (I_{VW} \otimes U_{\mc N})\rho_{VWA} (I_{VW} \otimes U_{\mc N}^{\dagger}).
\end{equation}
The quantum system $\mc H_W$ is usually treated as a conditioning system. When $\mc H_V$ is intended to be a classical system, it is replaced by $\mc X$, which is a finite set.

\begin{definition}\label{def:weaker degradability}
    The following are progressively weaker notions of degradability for a channel $\mc N$:
    \begin{enumerate}
        \item $\mc N$ is degradable, if there exists another quantum channel $\mc D$ such that $\mc D \circ \mc N = \mc N^c$.
        \item $\mc N$ is completely informationally degradable, if for any quantum systems $\mc H_V,\mc H_W$ and tripartite quantum state $\rho_{VWA}$ supported on $\mc H_V\otimes \mc H_W \otimes \mc H_A$, we have 
        \begin{equation}
            I(V;B|W)_{\rho_{VWB}} \ge I(V;E|W)_{\rho_{VWE}},
        \end{equation}
        where the conditional mutual information is defined as $I(V;B|W):= I(V;BW) - I(V;W)$.
        \item $\mc N$ is completely less noisy, if for any classical system $\mc X$, any quantum systems $\mc H_W$ and classical-quantum state $\rho_{\mc X WA} = \sum_{x \in \mc X} p(x) \ketbra{x}\otimes \rho_{WA}^x$, we have 
        \begin{equation}
            I(\mc X;B|W)_{\rho_{\mc X WB}} \ge I(\mc X;E|W)_{\rho_{\mc X WE}} .
        \end{equation}
        \item $\mc N$ is informationally degradable, if for any quantum system $\mc H_V$ and bipartite quantum state $\rho_{VA}$ supported on $\mc H_V \otimes \mc H_A$, we have 
        \begin{equation}
            I(V;B)_{\rho_{VB}} \ge I(V;E)_{\rho_{VE}}.
        \end{equation}
        \item $\mc N$ is less noisy, if for any classical system $\mc X$ and classical-quantum state $\rho_{\mc X A} = \sum_{x \in \mc X} p(x) \ketbra{x}\otimes \rho_{A}^x$, we have 
        \begin{equation}
            I(\mc X;B)_{\rho_{\mc XB}} \ge I(\mc X;E)_{\rho_{\mc XE}}.
        \end{equation}
    \end{enumerate}
\end{definition}
As shown by \cite{hirche2022contraction}, the above degradability notions form a hierarchy, with each notion being weaker than the preceding one. An additional parallel notion called \textit{regularized less noisy} was also introduced in \cite{watanabe2012private, hirche2022contraction}: $\mc N$ is \textit{regularized less noisy}, if for any $n \ge 1$ and any classical system $\mc X$ and classical-quantum state $\rho_{\mc X A^n} = \sum_{x \in \mc X} p(x) \ketbra{x}\otimes \rho_{A^n}^x$, where $\rho_{A^n}^x$ are quantum states on $A^{\otimes n}$, we have 
\begin{equation}
     I(\mc X;B^n)_{\rho_{\mc X B^n}} \ge I(\mc X;E^n)_{\rho_{\mc X E^n}}.
\end{equation}
\begin{figure*}
    \centering\includegraphics[width=1.0\textwidth]{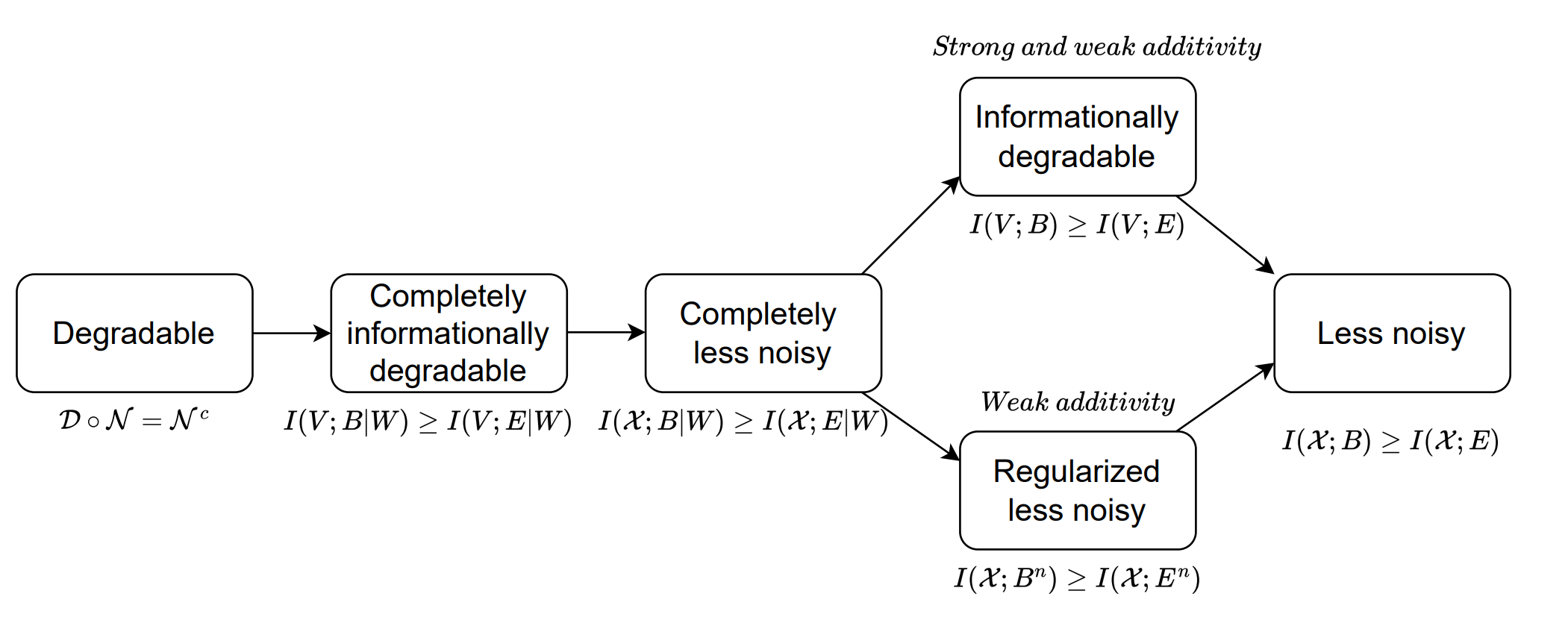}
    \caption{Hierarchy of notions of weaker degradability.}
    \label{fig:hierarchy}
\end{figure*}

\added{It remains unclear whether informational degradability or the regularized less-noisy property is the stronger condition. Even so, both notions provide sufficient criteria for the additivity of quantum capacities. In particular, Watanabe \cite{watanabe2012private} proved weak additivity whenever the channel is regularized less noisy, and Cross \emph{et al.} \cite{cross2017uniform} established weak additivity for informationally degradable channels. Building on the latter, we now extend the argument to demonstrate strong additivity for informationally degradable channels. For the convenience of the reader, we summarize all the results in the following theorem:}
\begin{theorem}\label{thm:weak degradability}
\begin{enumerate}
    \item If $\mc N$ is regularized less noisy, then for any $n \geq 1$:
    \[
    \mc Q^{(1)}(\mc N^{\otimes n}) = n \mc Q^{(1)}(\mc N).
    \]
    \item If $\mc N$ and $\mc M$ are informationally degradable channels, then for any $n, m \geq 1$:
    \[
    \mc Q^{(1)}(\mc N^{\otimes n} \otimes \mc M^{\otimes m}) = n \mc Q^{(1)}(\mc N) + m \mc Q^{(1)}(\mc M).
    \]
    In particular, this implies:
    \[
    \mc Q^{(1)}(\mc N^{\otimes n}) = n \mc Q^{(1)}(\mc N),
    \]
    for any informationally degradable channel $\mc N$.
\end{enumerate}
\end{theorem}
\begin{proof}
    1) The key ingredient is the divergence contraction property proved in \cite[Proposition 4]{watanabe2012private} and \cite[Proposition 2.3]{hirche2022contraction}: suppose $\mc N^{A \to B}$ and $\mc M^{A \to \wt B}$ are two quantum channels generated by isometries $U_{\mc N}: \mc H_A \to \mc H_B \otimes \mc H_E, \ U_{\mc M}: \mc H_A \to \mc H_{\wt B}\otimes \mc H_{\wt E}$ and $\eta \ge 0$. Then the following the properties are equivalent:
    \begin{itemize}
        \item For any classical-quantum state $\rho_{\mc X A} = \sum_{x \in \mc X} p(x) \ketbra{x}\otimes \rho_{A}^x$, we have 
        \begin{equation}
            \eta I(\mc X;B)_{\rho_{\mc X B}} \ge I(\mc X; \wt B)_{\rho_{\mc X \wt B}}.
        \end{equation}
        \item For any state $\rho_A, \sigma_A$ with $\text{supp} \rho_A \subseteq \text{supp} \sigma_A$, we have 
        \begin{equation}
            \eta D(\mc N(\rho_A) \| \mc N(\sigma_A)) \ge D(\mc M(\rho_A) \| \mc M(\sigma_A)).
        \end{equation}
    \end{itemize}
    For regularized less noisy channel $\mc N$, given $n\ge 1$, denote the isometry of $\mc N^{\otimes n}$ as $U_{\mc N^{\otimes n}}: \bigotimes_{i=1}^n \mc H_{A_i} \to \bigotimes_{i=1}^n \mc H_{B_i} \otimes \bigotimes_{i=1}^n \mc H_{E_i}$. Then for any $n$-partite state $\rho_{A^n}$, denote $\sigma_{A^n} = \rho_{A_1}\otimes \cdots \otimes \rho_{A_n}$ where $\rho_{A_i}$ is the reduced state of $\rho_{A^n}$ on $i$-th system. Applying the above equivalent conditions for $\mc N^{\otimes}, (\mc N^c)^{\otimes n}, \eta = 1$, we have 
    \begin{equation}
    \begin{aligned}
                 D(\mc N^{\otimes n}(\rho_{A^n}) \| \mc N^{\otimes n}(\sigma_{A^n})) \ge D((\mc N^c)^{\otimes n}(\rho_{A^n}) \| (\mc N^c)^{\otimes n}(\sigma_{A^n})),
    \end{aligned}
    \end{equation}
    which by definition is equivalent to 
    \begin{equation}
    \begin{aligned}
        S(\rho_{B^n}) - S(\rho_{E^n}) & \le -\Tr(\rho_{B^n} \log (\rho_{B_1}  \otimes \cdots \otimes \rho_{B_n} ))  + \Tr(\rho_{E^n} \log (\rho_{E_1}  \otimes \cdots \otimes \rho_{E_n} )) \\
        & = \sum_{i=1}^n S(\rho_{B_i}) - S(\rho_{E_i}) \le n \mc Q^{(1)}(\mc N).
    \end{aligned}
    \end{equation}
    By choosing $\rho_{A^n}$ as the optimal state achieving $\mc Q^{(1)}(\mc N^{\otimes n})$, we show subadditivity thus additivity of coherent information. 

    2) The key ingredient is the following telescoping argument: Suppose $\rho$ is a state on $\bigotimes_{i=1}^n \mc H_{B_i} \otimes \bigotimes_{i=1}^n \mc H_{E_i}$, and denote $S(B) = S(\rho_B)$ for notational simplicity. Then 
\begin{equation}
    S(B_1\cdots B_n)-S(E_1\cdots E_n) = \sum_{i=1}^n S(B_iV_i) - S(E_iV_i),
\end{equation}
where $V_i$ is defined as
\begin{equation}
    V_i = \begin{cases}
    B_2\cdots B_n,\ i = 1, \\
    E_1 \cdots E_{i-1} B_{i+1}\cdots B_n,\ 2\le i \le n-1, \\
    E_1 \cdots E_{n-1},\ i=n.
    \end{cases}
\end{equation}
This follows by adding and subtracting $S(E_1\cdots E_j B_{j+1} \cdots B_n)$ for $1 \le j \le n-1$ and rearranging the $n$ terms.

Back to the proof of strong additivity for informationally degradable channels, we denote the isometry of informationally degradable channels as
\begin{align*}
     U_{\mc N}: \mc H_A \to \mc H_B \otimes \mc H_E, \ U_{\mc M}: \mc H_{\wt A} \to \mc H_{\wt B}\otimes \mc H_{\wt E}.
\end{align*}
Then for any input state $\rho_{A^n\wt A^n}$, we denote 
\begin{align*}
    & \rho_{B_1\cdots B_nE_1\cdots E_n \wt B_1\cdots \wt B_m \wt E_1\cdots \wt E_m}  =( U_{\mc N}^{\otimes n} \otimes U_{\mc M}^{\otimes m} ) \rho_{A^n\wt A^n} ( (U_{\mc N}^{\dagger})^{\otimes n} \otimes (U_{\mc M}^{\dagger})^{\otimes m} ).
\end{align*}
Our goal is to prove subadditivity:
\begin{equation}
    \begin{aligned}
        & S(B_1\cdots B_n \wt B_1\cdots \wt B_m) - S(E_1\cdots E_n \wt E_1\cdots \wt E_m) \le \sum_{i=1}^n (S(B_i) - S(E_i)) + \sum_{j=1}^m (S(\wt B_j) - S(\wt E_j)).
    \end{aligned}
\end{equation}
Applying the above telescoping lemma, we have 
    \begin{align*}
       & S(B_1\cdots B_n \wt B_1\cdots \wt B_m) - S(E_1\cdots E_n \wt E_1\cdots \wt E_m) \\
       & = S(B_1\cdots B_n \wt B_1\cdots \wt B_m) - S(E_1\cdots E_n \wt B_1\cdots \wt B_m)  + S(E_1\cdots E_n \wt B_1\cdots \wt B_m) - S(E_1\cdots E_n \wt E_1\cdots \wt E_m) \\
 & = \sum_{i=1}^n S(B_i V_i \wt B_1\cdots \wt B_m ) - S(E_i V_i \wt B_1\cdots \wt B_m ) + \sum_{j=1}^m S(E_1\cdots E_n \wt B_j \wt V_j) - S(E_1\cdots E_n \wt E_j \wt V_j).
    \end{align*}
    Note that informationally degradability $I(V;B) \ge I(V;E)$ implies 
    \begin{equation}
            S(B)-S(E) \ge S(BV) - S(EV)
    \end{equation}
for any quantum system $\mc H_V$. Then the above inequality proceeds as 
\begin{equation}
\begin{aligned}
        & S(B_1\cdots B_n \wt B_1\cdots \wt B_m) - S(E_1\cdots E_n \wt E_1\cdots \wt E_m) \le \sum_{i=1}^n S(B_i) - S(E_i) + \sum_{j=1}^m S(\wt B_j) - S(\wt E_j),
\end{aligned}
\end{equation}
which concludes the proof.
\end{proof}
\begin{remark}
    It is not clear whether $\mc N \otimes \mc N$ is also informationally degradable if $\mc N$ is informationally degradable. 
\end{remark}
\begin{remark}
Recall from \cite[Lemma 1, Theorem 6]{smith2008quantumss} the definition of the \emph{symmetric side-channel-assisted capacity}:
\begin{align*}
\mathcal{Q}_{ss}(\mathcal{N}) & = \mathcal{Q}_{ss}^{(1)}(\mathcal{N}) = \sup_{\mathcal{H}_V,\mathcal{H}_W} \ \sup_{\rho_{VWA}} \bigl(I(V;B \mid W) - I(V;E \mid W)\bigr) \\
&= \sup_{\mathcal{A}\ \text{symmetric}} \mathcal{Q}^{(1)}\bigl(\mathcal{N} \otimes \mathcal{A}\bigr).
\end{align*}
It is easy to see that for informationally degradable channel $\mc N$, $\mc Q_{ss}(\mc N) = \mc Q(\mc N) = \mc Q^{(1)}(\mc N)$. Moreover, for completely informationally degradable channel $\mc M$, we have $\mc Q_{ss}(\mc M^c) = \mc Q(\mc M^c) = 0$. 
\end{remark}

A possible example separating different notions in Definition \ref{def:weaker degradability} is given as follows. Suppose $\mc N^{A \to B_1}, \mc M^{A \to B_2}$ are two degradable channels with the same input. Define 
\begin{equation}\label{flagged channel:general}
    \Psi_{p,\mc N,\mc M}:= p \ketbra{0}{0} \otimes \mc N^{A \to B_1} + (1-p) \ketbra{1}{1}\otimes (\mc M^{A \to B_2})^c, 
\end{equation}
which is a flagged channel mixed by degradable and anti-degradable channels. Intuitively, when $p$ is close to 1, the channel behaves like a degradable channel, but the anti-degradable part can destroy the degradability, while the weaker notions of degradability in Definition \ref{def:weaker degradability} may be preserved. 

We denote the isometries generating $\mc N,\mc M$ as 
\begin{equation}
    U_{\mc N}: \mc H_A \to \mc H_{B_1} \otimes \mc H_{E_1},\quad U_{\mc M}: \mc H_A \to \mc H_{B_2} \otimes \mc H_{E_2}
\end{equation}
Following the notation in Definition \ref{def:weaker degradability}, we introduce the following non-negative constants (we adopt the convention $\frac{0}{0} = 1$):
\begin{align}
    & R_1(\mc N,\mc M):= \inf_{\mc H_V,\mc H_W} \inf_{\rho_{VWA}} \frac{I(V;B_1|W) - I(V;E_1|W)}{I(V;B_2|W) - I(V;E_2|W)}, \label{ratio:1}\\
    & R_2(\mc N,\mc M):= \inf_{\mc X,\mc H_V} \inf_{\rho_{\mc X WA}} \frac{I(\mc X;B_1|W) - I(\mc X;E_1|W)}{I(\mc X;B_2|W) - I(\mc X;E_2|W)}, \label{ratio:2}\\
    & R_3(\mc N,\mc M):= \inf_{\mc H_V} \inf_{\rho_{VA}} \frac{I(V;B_1) - I(V;E_1)}{I(V;B_2) - I(V;E_2)}, \label{ratio:3}\\
    & \wt R_3(\mc N,\mc M):= \inf_{n \ge 1,\mc X} \inf_{\rho_{\mc XA^n}} \frac{I(\mc X;B_1^n) - I(\mc X;E^n_1)}{I(\mc X;B^n_2) - I(\mc X;E^n_2)}, \label{ratio:3'}\\
    & R_4(\mc N,\mc M):= \inf_{\mc X} \inf_{\rho_{\mc X A}} \frac{I(\mc X;B_1) - I(\mc X;E_1)}{I(\mc X;B_2) - I(\mc X;E_2)}.\label{ratio:4}
\end{align}
Since $\mc N,\mc M$ are degradable, the data processing inequality ensures that $R_i \geq 0$ for $1 \leq i \leq 4$. Furthermore, due to the hierarchical structure of the notions in Definition \ref{def:weaker degradability}, we have
\begin{align*}
    R_1 \leq R_2 \leq R_3,\ \wt{R}_3 \leq R_4.
\end{align*}
The key question is whether these constants are strictly positive. While optimization over arbitrary quantum systems leaves it unclear whether the infimum can approach zero, there are indications that positivity is plausible for degradable channels. For example, entangled states are known not to enhance the coherent information of degradable channels, which supports the conjecture that $R_i > 0$ for some $\mc N, \mc M$. This issue is closely tied to the challenge of determining quantum dimension bounds, as explored in \cite{beigi2013dimension, hirche2020alphabet}, a notoriously difficult task. Providing a rigorous justification for the positivity of these constants in general remains an open question. For positivity of $R_4$, we refer the reader to \cite{BGSW24}. For the positivity of $R_3$, we provide some justifications in Section \ref{sec:flagged channel example}.

Using the property for flagged channels, we can show the following:
\begin{prop}\label{main:less noisy}
    Suppose $\Psi_{p,\mc N,\mc M}$ is defined by the flagged mixture of degradable and anti-degradable channels in \eqref{flagged channel:general}, then
    \begin{itemize}
        \item If $p \ge \frac{1}{1+R_1(\mc N,\mc M)}$, $\Psi_{p,\mc N,\mc M}$ is completely informational degradable.
        \item If $p \ge \frac{1}{1+R_2(\mc N,\mc M)}$, $\Psi_{p,\mc N,\mc M}$ is completely less noisy.
        \item If $p \ge \frac{1}{1+R_3(\mc N,\mc M)}$, $\Psi_{p,\mc N,\mc M}$ is informational degradable.
        \item If $p \ge \frac{1}{1+\wt R_3(\mc N,\mc M)}$, $\Psi_{p,\mc N,\mc M}$ is regularized less noisy.
        \item If $p \ge \frac{1}{1+R_4(\mc N,\mc M)}$, $\Psi_{p,\mc N,\mc M}$ is less noisy.
    \end{itemize}
\end{prop}
\begin{proof}
We demonstrate the case of completely informational degradability; the other cases follow similarly. Note that the (conditional) mutual information under convex combination of orthogonal states is additive, see Lemma \ref{lemma:basic}, we have 
\begin{equation}
\begin{aligned}
 I(V;B|W) - I(V;E|W) &= p (I(V;B_1|W) - I(V;E_1|W)) - (1-p) (I(V;B_2|W) - I(V;E_2|W)),
\end{aligned}
\end{equation}
where $\mc H_B = \mc H_{B_1}\oplus \mc H_{B_2},\ \mc H_E = \mc H_{E_1}\oplus \mc H_{E_2}$.

Therefore, $I(V;B|W) - I(V;E|W) \ge 0$ is equivalent to 
\begin{align}\label{eqn: key inequality of less noisy}
    \frac{I(V;B_1|W) - I(V;E_1|W)}{I(V;B_2|W) - I(V;E_2|W)} \ge \frac{1-p}{p}.
\end{align}
If $p \ge \frac{1}{1+R_1(\mc N,\mc M)}$, \eqref{eqn: key inequality of less noisy} holds, which concludes the proof.
\end{proof}
In Section \ref{sec:flagged channel example}, we present an example of $\mc N,\mc M$ such that as long as $p<1$, the channel $\Psi_{p,\mc N,\mc M}$ is neither degradable nor anti-degradable, which makes a separation of different notions in Definition \ref{def:weaker degradability} possible.

\subsection{Additivity via data processing inequality}
In this subsection, we discuss the construction of non-degradable channels from degradable ones in such a way that, while degradability is destroyed, the additivity property of the quantum capacity is preserved. The following notion is 
\added{\begin{definition}\label{def:weak_dom}
A quantum channel $\mc N$ is said to be weakly dominated by another quantum channel $\mc M$ if 
\begin{equation}
    \mc Q^{(1)}(\mc N) \le \mc Q^{(1)}(\mc M),
\end{equation}
where $\mc Q^{(1)}$ is defined in \eqref{def:max coherent information}.
\end{definition}}
Using the above notion, \added{we demonstrate when the additivity properties of $\mc N$ can be inherited from $\widehat{\mc N}$ under additional assumptions.} The following theorem is implicit in \cite{smith2008additive}:
\begin{theorem}\label{thm:additivity via DPI}
Let $\mc N$ and $\widehat{\mc N}$ be quantum channels. \added{Suppose 
\begin{itemize}
    \item $\mc N$ can be simulated by $\widehat{\mc N}$, i.e., there exist quantum channels $\mc E$ and $\mc D$ such that
\begin{equation}\label{eqn:simulation}
    \mc D \circ \widehat{\mc N} \circ \mc E = \mc N.
\end{equation}
    \item $\widehat{\mc N}$ is weakly dominated by $\mc N$.
\end{itemize}}
Then the following properties hold:
\begin{itemize}
    \item If $\mc Q(\widehat{\mc N}) = \mc Q^{(1)}(\widehat{\mc N})$, then $\mc Q(\mc N) = \mc Q^{(1)}(\mc N)$.
    \item If $\mc Q^{(1)}(\widehat{\mc N} \otimes \mc M) = \mc Q^{(1)}(\widehat{\mc N}) + \mc Q^{(1)}(\mc M)$ for some other channel $\mc M$, then 
    \begin{equation}
        \mc Q^{(1)}(\mc N \otimes \mc M) = \mc Q^{(1)}(\mc N) + \mc Q^{(1)}(\mc M).
    \end{equation}
\end{itemize}
\end{theorem}
\begin{proof}
Using Bottleneck inequality \eqref{bottleneck inequality} and \eqref{eqn:simulation}, we have
\begin{equation}
    \mc Q^{(1)}(\widehat{\mc N}) \ge \mc Q^{(1)}(\mc N),\quad \mc Q(\widehat{\mc N}) \ge \mc Q(\mc N).
\end{equation}
Therefore, we have 
\begin{equation}
    \mc Q^{(1)}(\widehat{\mc N}) = \mc Q(\widehat{\mc N}) \ge \mc Q(\mc N) \ge \mc Q^{(1)}(\mc N).
\end{equation}
\added{Since $\widehat{\mc N}$ is weakly dominated by $\mc N$, we have $\mc Q^{(1)}(\mc N) \ge \mc Q^{(1)}(\widehat{\mc N})$, thus we get the weak additivity.} Using similar argument, we have
\begin{align*}
 \mc Q^{(1)}(\widehat{\mc N}) + \mc Q^{(1)}(\mc M) & = \mc Q^{(1)}(\widehat{\mc N} \otimes  \mc M) \\
    & \ge \mc Q^{(1)}(\mc N\otimes  \mc M) \\
    & \ge \mc Q^{(1)}(\mc N) + \mc Q^{(1)}(\mc M) \\
    & \ge \mc Q^{(1)}(\widehat{\mc N}) + \mc Q^{(1)}(\mc M).
\end{align*}
Therefore, the inequalities collapse into equalities, implying $\mc Q^{(1)}(\mc N\otimes  \mc M) = \mc Q^{(1)}(\mc N) + \mc Q^{(1)}(\mc M)$.
\end{proof}
Using the above theorem and direct sum of channels, one can immediately get a class of channels which are neither degradable nor anti-degradable, but this class has additive coherent information:
\begin{corollary}\label{cor:direct sum}
    Suppose $\Phi_1$ is degradable and $\Phi_2$ is anti-degradable. Then $\Phi_1 \oplus \Phi_2$ defined by \eqref{def:direct sum channel} satisfies:
    \begin{itemize}
        \item $\mc Q(\Phi_1 \oplus \Phi_2) = \mc Q^{(1)}(\Phi_1 \oplus \Phi_2)$. 
        \item For any degradable channel $\Psi$, we have 
        \begin{equation}
            \mc Q^{(1)}\left( (\Phi_1 \oplus \Phi_2) \otimes \Psi\right) = \mc Q^{(1)}(\Phi_1 \oplus \Phi_2) + \mc Q^{(1)}(\Psi). 
        \end{equation}
    \end{itemize}
\end{corollary}
To prove the corollary, the following well-known fact of anti-degradable channels is useful:
\begin{lemma}\label{lemma:anti-degradable decomposition}
    Suppose $\mc N: \mb B(\mc H_A) \to \mb B(\mc H_B)$ is anti-degradable, then there exists a symmetric channel $\widehat{\mc N}: \mb B(\mc H_A) \to \mb B(\widehat{\mc H_B})$ and a quantum channel $\mc P:\mb B(\widehat{\mc H_B}) \to \mb B(\mc H_B)$ such that 
    \begin{equation}
        \mc N = \mc P \circ\widehat{\mc N}. 
    \end{equation}
\end{lemma}
\begin{proof}
The construction of the symmetric channel $\widehat{\mc N}$ is given by a flagged channel with equal probability:
\begin{align*}
    \widehat{\mc N} = \frac{1}{2} \ketbra{0}{0} \otimes \mc N + \frac{1}{2} \ketbra{1}{1} \otimes \mc N^c.
\end{align*}
This channel is symmetric, since
\begin{align*}
     \widehat{\mc N}^c = \frac{1}{2} \ketbra{0}{0} \otimes \mc N^c + \frac{1}{2} \ketbra{1}{1} \otimes \mc N,
\end{align*}
and the degrading and anti-degrading map is constructed via swapping the flag. The processing channel $\mc P$ is constructed by recovering $\mc N$ from $\mc N^c$ (implied by the fact that $\mc N$ is anti-degradable) and tracing out the flag. 
\end{proof}

\begin{proof}[Proof of Corollary \ref{cor:direct sum}]
    Via Lemma \ref{lemma:anti-degradable decomposition}, there exist a symmetric channel $\widehat{\Phi_2}$ and another quantum channel $\mc P$ such that $\Phi_2 = \mc P \circ \widehat{\Phi_2}$. Therefore, $\Phi_1 \oplus \Phi_2$ can be simulated by $\Phi_1 \oplus \widehat{\Phi_2}$, i.e., 
    \begin{align*}
        (id \oplus \mc P )\circ (\Phi_1 \oplus \widehat{\Phi_2}) = \Phi_1 \oplus \Phi_2. 
    \end{align*}
    Moreover, $\Phi_1 \oplus \widehat{\Phi_2}$ is weakly dominated by $\Phi_1 \oplus \Phi_2$:
    \begin{align*}
        & \mc Q^{(1)}(\Phi_1 \oplus \widehat{\Phi_2}) = \max\{\mc Q^{(1)}(\Phi_1), \mc Q^{(1)}(\widehat{\Phi_2})\} =  \mc Q^{(1)}(\Phi_1) \le \mc Q^{(1)}(\Phi_1 \oplus \Phi_2),
    \end{align*}
    \added{where the first equality follows from Lemma \ref{lemma:basic}; the second equality follows from $\mc Q^{(1)}(\widehat{\Phi_2}) = 0$ since $\widehat{\Phi_2}$ is symmetric (in particular anti-degradable) and thus has zero capacity; and the inequality follows from the definition of maximal coherent information and choosing the optimized states in the first component of the direct sum. }
    
    Finally, note that the additivity properties for $\Phi_1 \oplus \widehat{\Phi_2}$ follows from the fact that it is a degradable channel. We conclude the proof via Theorem \ref{thm:additivity via DPI}. 
\end{proof}
In the remaining section, we show that the additivity examples discussed in \cite{Gao_2018, chessa2021quantum, chessa2023resonant} can be shown via Theorem \ref{thm:additivity via DPI}. 
\begin{example}\label{example:Gao}
(\cite[Section 4.3]{Gao_2018}) The quantum channel 
    \begin{equation}\label{example:dephasing}
    \begin{aligned}
       & \Phi_{\alpha}\begin{pmatrix}
            a_{00} & a_{01} & a_{02} & a_{03} \\
            a_{10} & a_{11} & a_{12} & a_{13} \\
            a_{20} & a_{21} & a_{22} & a_{23} \\
            a_{30} & a_{31} & a_{32} & a_{33} 
        \end{pmatrix}  = \begin{pmatrix}
            a_{00} + a_{11} & \alpha a_{02} & \alpha a_{13} \\
            \alpha a_{20} & a_{22} & 0 \\
            \alpha a_{31} & 0 & a_{33} 
        \end{pmatrix},\quad |\alpha| \le 1,
    \end{aligned}
\end{equation}
has weak additive coherent information, which is neither degradable nor anti-degradable if $\alpha \in (0,1]$.
\end{example}
Note that in \cite{Gao_2018}, the authors developed an upper bound of $\mc Q(\Phi_{\alpha})$ via comparison to TRO channels. Here, we present a simpler proof using Theorem \ref{thm:additivity via DPI}: 
\begin{proof}
    Denote $\mc D_{\alpha}$ as the qubit dephasing channel defined by 
    \begin{align*}
        \mc D_{\alpha} \begin{pmatrix}
            a_{00} & a_{01} \\
            a_{10} & a_{11} 
        \end{pmatrix} = \begin{pmatrix}
            a_{00} & \alpha a_{01} \\
            \alpha a_{10} & a_{11} 
        \end{pmatrix}.
    \end{align*}
    Via \eqref{def:direct sum channel}, $\mc D_\alpha \oplus \mc D_\alpha$ is given by 
    \begin{align*}
        (\mc D_\alpha \oplus \mc D_\alpha)\begin{pmatrix}
            a_{00} & a_{01} & a_{02} & a_{03} \\
            a_{10} & a_{11} & a_{12} & a_{13} \\
            a_{20} & a_{21} & a_{22} & a_{23} \\
            a_{30} & a_{31} & a_{32} & a_{33} 
        \end{pmatrix} = \begin{pmatrix}
            a_{00} & \alpha a_{01} & 0 & 0 \\
            \alpha a_{10} & a_{11} & 0 & 0 \\
            0  & 0 & a_{22} & \alpha a_{23} \\
            0 & 0 & \alpha a_{32} & a_{33} 
        \end{pmatrix}
    \end{align*}
    Therefore, $\Phi_{\alpha}$ can be simulated by $\mc D_{\alpha} \oplus \mc D_{\alpha}$ in the sense of \eqref{eqn:simulation}, i.e., there exist channels $\mc E,\mc D$ such that $\mc D \circ (\mc D_{\alpha} \oplus \mc D_{\alpha})\circ \mc E = \Phi_\alpha$. To be more specific, $\mc E,\mc D$ are given by 
    \begin{align*}
        & \mc E\begin{pmatrix}
            a_{00} & a_{01} & a_{02} & a_{03} \\
            a_{10} & a_{11} & a_{12} & a_{13} \\
            a_{20} & a_{21} & a_{22} & a_{23} \\
            a_{30} & a_{31} & a_{32} & a_{33} 
        \end{pmatrix} = \begin{pmatrix}
            1 & 0 & 0 & 0 \\
            0 & 0 & 1 & 0 \\
            0 & 1 & 0 & 0 \\
            0 & 0 & 0 & 1 
        \end{pmatrix} \begin{pmatrix}
            a_{00} & a_{01} & a_{02} & a_{03} \\
            a_{10} & a_{11} & a_{12} & a_{13} \\
            a_{20} & a_{21} & a_{22} & a_{23} \\
            a_{30} & a_{31} & a_{32} & a_{33} 
        \end{pmatrix}\begin{pmatrix}
            1 & 0 & 0 & 0 \\
            0 & 0 & 1 & 0 \\
            0 & 1 & 0 & 0 \\
            0 & 0 & 0 & 1 
        \end{pmatrix}=\begin{pmatrix}
            a_{00} & a_{02} & a_{01} & a_{03} \\
            a_{20} & a_{22} & a_{21} & a_{23} \\
            a_{10} & a_{12} & a_{11} & a_{13} \\
            a_{30} & a_{32} & a_{31} & a_{33} 
        \end{pmatrix}, \\
        & \mc D \begin{pmatrix}
            a_{00} & a_{01} & a_{02} & a_{03} \\
            a_{10} & a_{11} & a_{12} & a_{13} \\
            a_{20} & a_{21} & a_{22} & a_{23} \\
            a_{30} & a_{31} & a_{32} & a_{33} 
        \end{pmatrix} = \begin{pmatrix}
            a_{00} + a_{22} & a_{01} & a_{23} \\
            a_{10} & a_{11} & 0 \\
            a_{32} & 0 & a_{33}
        \end{pmatrix}
    \end{align*}
    Via Corollary \ref{cor:direct sum}, $\mc D_{\alpha} \oplus \mc D_{\alpha}$ has weak additive coherent information. Moreover, for $\alpha \in [0,1]$, it is well-known that for dephasing channel,  
    \begin{align*}
        \mc Q^{(1)}(\mc D_{\alpha}\oplus \mc D_\alpha) = \mc Q^{(1)}(\mc D_{\alpha}) = 1 - h\left( \frac{1+\alpha}{2} \right),
    \end{align*}
    where \(h(x)\) is the binary entropy function. By choosing the input state \(\frac{1}{2} \ketbra{0}{0} + \frac{1}{2} \ketbra{2}{2}\) for \(\Phi_{\alpha}\), we have:
\[
\mc Q^{(1)}(\Phi_{\alpha}) \geq 1 - h\left( \frac{1+\alpha}{2} \right) \quad \text{if } \alpha \in [0,1],
\]
showing that $\mc D_{\alpha} \oplus \mc D_{\alpha}$ is weakly dominated by $\Phi_\alpha$. We conclude the proof by applying Theorem \ref{thm:additivity via DPI}.
\end{proof}
In \cite{chessa2021quantum,chessa2023resonant}, several examples of channels are presented whose coherent information is additive. Here, we observe that every such channel that is neither degradable nor anti-degradable arises as a special case of Theorem~\ref{thm:additivity via DPI}. We discuss an example called multi-level amplitude damping channel in \cite{chessa2023resonant}:
\begin{example}
The qutrit channel 
\begin{align*}
    & \mc A_{\gamma_0,\gamma_1}\left(\begin{pmatrix}
            a_{00} & a_{01} & a_{02} \\
            a_{10} & a_{11} & a_{12} \\
            a_{20} & a_{21} & a_{22} 
        \end{pmatrix}\right)  = \begin{pmatrix}
            a_{00}+\gamma_0 a_{22} & a_{01} & \sqrt{\gamma_2}a_{02} \\
            a_{10} & a_{11}+ \gamma_1 a_{22} & \sqrt{\gamma_2} a_{12} \\
            \sqrt{\gamma_2}a_{20} & \sqrt{\gamma_2}a_{21} & \gamma_2 a_{22} 
        \end{pmatrix}
\end{align*}
where $\gamma_2:= 1-\gamma_0 - \gamma_1$ and $0\le \gamma_0,\gamma_1 \le \gamma_0+\gamma_1 \le 1$, is neither degradable nor anti-degradable when $\gamma_0+ \gamma_1 > \frac{1}{2}$, but has weak additive coherent information.
\end{example}
\begin{proof}
    The isometry $U_{\mc A_{\gamma_0,\gamma_1}}: \mc H_A \to \mc H_B \otimes \mc H_E$ generating $\mc A_{\gamma_0,\gamma_1}$ is given by
    \begin{align*}
        & U_{\mc A_{\gamma_0,\gamma_1}} \ket{0} = \ket{0}\otimes \ket{0}, \\
        & U_{\mc A_{\gamma_0,\gamma_1}} \ket{1} = \ket{1}\otimes \ket{0}, \\
        & U_{\mc A_{\gamma_0,\gamma_1}} \ket{2} = \sqrt{\gamma_2}\ket{2}\otimes \ket{0} +\sqrt{\gamma_1}\ket{1}\otimes \ket{1} + \sqrt{\gamma_0}\ket{0}\otimes \ket{2}.
    \end{align*}
    This channel can be intuitively understood as a level-3 amplitude damping channel, such that $\ket{0},\ket{1}$ are fixed, and $\ket{2}$ will be damped to $\ket{0} $ or $ \ket{1}$ with probability $\gamma_0$ or $\gamma_1$. It is straightforward to check that if $\gamma_2 = 1-\gamma_0-\gamma_1 \ge \frac{1}{2}$, the channel is degradable, but is neither degradable nor anti-degradable when $\gamma_2 < \frac{1}{2}$. However, if $\gamma_2 < \frac{1}{2}$, there exists $\gamma_0',\gamma_1'$ with $\gamma_0' + \gamma_1' = \frac{1}{2}$ such that $\mc A_{\gamma_0,\gamma_1}$ can be simulated by $\mc A_{\gamma_0',\gamma_1'}$, see \cite[Appendix A]{chessa2023resonant}. Moreover, since $\mc A_{\gamma_0',\gamma_1'}$ is degradable, it is straightforward to check that $\mc Q(\mc A_{\gamma_0',\gamma_1'}) = \mc Q^{(1)}(\mc A_{\gamma_0',\gamma_1'}) = 1$. Moreover, since the two by two upper left block is perfectly transmitted, $\mc Q^{(1)}(\mc A_{\gamma_0,\gamma_1}) \ge 1 = \mc Q^{(1)}(\mc A_{\gamma_0',\gamma_1'})$. By Theorem \ref{thm:additivity via DPI}, we have $\mc Q(\mc A_{\gamma_0,\gamma_1}) = \mc Q^{(1)}(\mc A_{\gamma_0,\gamma_1}) = 1$
\end{proof}

\section{Analysis of flagged mixture of degradable and anti-degradable channels}\label{sec:flagged channel example}
In this section, we examine a class of flagged channels of the form \eqref{flagged channel:general}. Whenever \(p < 1\), these channels are neither degradable nor anti-degradable, ensuring that the various notions in Definition~\ref{def:weaker degradability} can be distinguished, provided that the nonnegative constants in \eqref{ratio:1}--\eqref{ratio:4} remain strictly positive.

Specifically, we consider flagged mixtures of qubit amplitude-damping channels. Let \(\mathcal{A}_\gamma\), with \(\gamma \in [0,1]\), denote the qubit amplitude-damping channel with the isometry 
\[
U_{\mathcal{A}_{\gamma}}\colon \mathcal{H}_A \to \mathcal{H}_B \otimes \mathcal{H}_E
\]
where \(\mathcal{H}_A = \mathcal{H}_B = \mathcal{H}_E = \mathbb{C}^2\) is given by
\begin{equation*}
\begin{aligned}
    U_{\mathcal{A}_{\gamma}} \ket{0} &= \ket{00}, \\
    U_{\mathcal{A}_{\gamma}} \ket{1} &= \sqrt{1-\gamma}\ket{10} + \sqrt{\gamma}\,\ket{01}.
\end{aligned}
\end{equation*}
We then form the complementary pair of quantum channels \((\Phi_{p,\gamma,\eta}, \Phi_{p,\gamma,\eta}^c)\) as follows:
\begin{equation}\label{eqn:complementary pair of flagged mixture}
\begin{aligned}
    \Phi_{p,\gamma,\eta}(\rho)
    &= (1-p)\,\ketbra{0} \otimes \mathcal{A}_{\gamma}(\rho)
       \;+\; p\,\ketbra{1} \otimes \mathcal{A}_{\eta}(\rho), \\[6pt]
    \Phi_{p,\gamma,\eta}^c(\rho)
    &= (1-p)\,\ketbra{0} \otimes \mathcal{A}_{\gamma}^c(\rho)
       \;+\; p\,\ketbra{1} \otimes \mathcal{A}_{\eta}^c(\rho),
\end{aligned}
\end{equation}
where $p,\gamma,\eta \in [0,1]$.

We first determine the parameter region \((p,\gamma,\eta)\) in which \(\Phi_{p,\gamma,\eta}\) is (non-)degradable and (non-)anti-degradable. Then, for parameter choices that make \(\Phi_{p,\gamma,\eta}\) neither degradable nor anti-degradable, we present evidence that for each \((\gamma,\eta)\) in this region, there exists a threshold \(p^*(\gamma,\eta)\) beyond which \(\Phi_{p,\gamma,\eta}\) becomes either informationally degradable or informationally anti-degradable (as in Definition~\ref{def:weaker degradability}). This threshold is given by Proposition~\ref{main:less noisy}.


\subsection{Degradable and anti-degradable regions}
We briefly illustrate the idea before we present the formal statement. First, when $\gamma, \eta$ are both less than or greater than $\frac{1}{2}$, $\mc A_{\gamma}, \mc A_{\eta}$ are degradable or anti-degradable, then it is well-known that flagged mixture of degradable or anti-degradable channels is again degradable or anti-degradable \cite{smith2008additive}. If one of $\gamma,\eta$ is strictly greater than $\frac{1}{2}$ and the other one is strictly smaller than $\frac{1}{2}$, then it is a flagged mixture of degradable and anti-degradable channels. In this case, a more general sufficient condition is given in \cite{fanizza2020quantum,kianvash2022bounding}, which does not present a full characterization of the non-degradable regions.

In our case, based on the special structure of the channels, we can characterize the full region of degradability or anti-degradability. Outside that region, it is neither degradable nor anti-degradable. The crucial idea is to construct crossing degrading maps as follows:
\begin{figure}[ht]
    \centering\includegraphics[width=.3\textwidth]{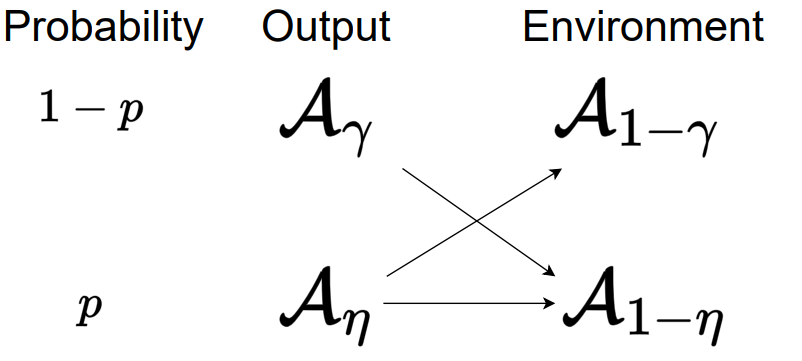}
    \caption{Construction of the degrading map for $p>\frac{1}{2}$, $\eta + \gamma \le 1$ and $\eta \le \frac{1}{2}$.}
    \label{fig:flagged degradable}
\end{figure}
\begin{prop}\label{main:degradable characterization}
    $\Phi_{p,\gamma,\eta}$ is degradable if and only if $(p,\gamma, \eta)$ satisfies one of the following conditions:
    \begin{enumerate}
        \item For $p = \frac{1}{2}$: $\eta + \gamma \le 1$. 
        \item For $p> \frac{1}{2}$: $\eta + \gamma \le 1$ and $\eta \le \frac{1}{2}$.
        \item For $p<\frac{1}{2}$: $\eta + \gamma \le 1$ and $\gamma \le \frac{1}{2}$.
    \end{enumerate}
    $\Phi_{p,\gamma,\eta}$ is anti-degradable if and only if $(p,\gamma, \eta)$ satisfies one of the following conditions:
    \begin{enumerate}
        \item For $p = \frac{1}{2}$: $\eta + \gamma \ge 1$. 
        \item For $p> \frac{1}{2}$: $\eta + \gamma \ge 1$ and $\eta \ge \frac{1}{2}$.
        \item For $p<\frac{1}{2}$: $\eta + \gamma \ge 1$ and $\gamma \ge \frac{1}{2}$.
    \end{enumerate}
\end{prop}
\begin{figure*}
    \centering\includegraphics[width=1.0\textwidth]{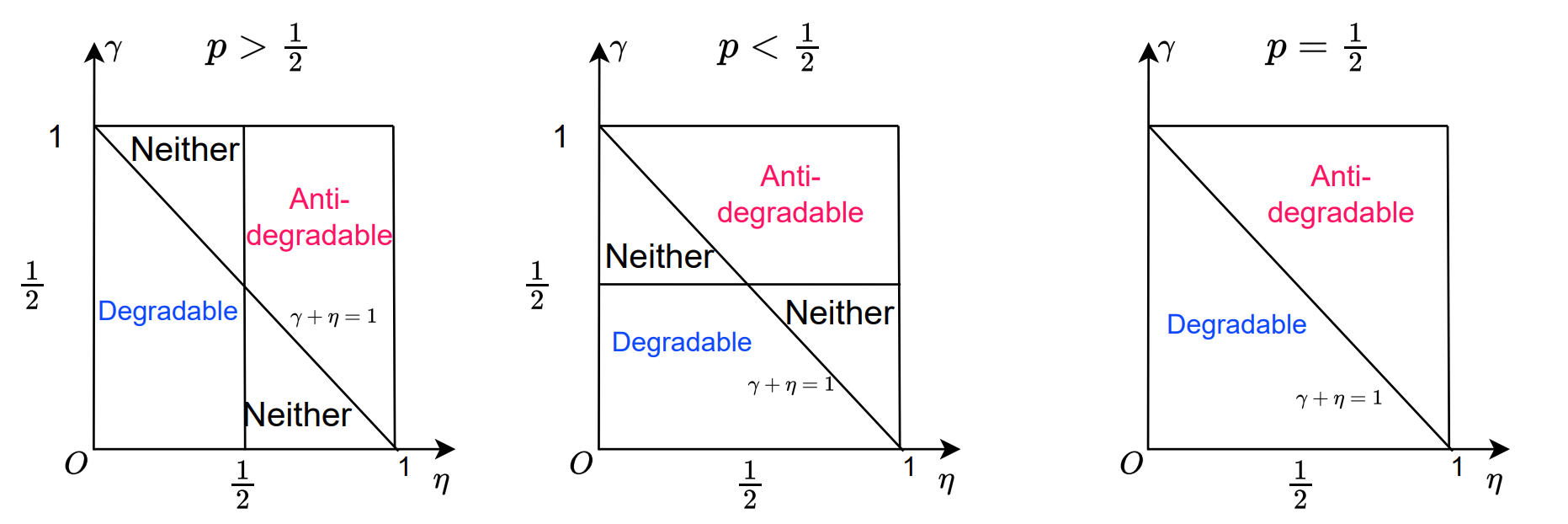}
    \caption{Degradable and anti-degradable regions for $\Phi_{p,\gamma,\eta}$.}
    \label{fig:deg region}
\end{figure*}
\begin{proof}
We only need to prove the degradable case, since by replacing $\gamma$ by $1-\gamma$ and $\eta$ by $1-\eta$, we get the anti-degradable region. Our proof is based on the well-known fact about the inversion and composition of qubit amplitude damping channel \cite{giovannetti2005information}: suppose $0\le \gamma_1.\gamma_2 < 1$, the inverse linear map $\mc A_{\gamma}^{-1}$ is unique and non-positive unless $\gamma = 0$. The explicit calculation of the inversion and composition is calculated as 
\begin{equation}\label{eqn:inversion of AD}
    \begin{aligned}
    & \mc A_{\gamma}^{-1}\begin{pmatrix}
        \rho_{00} & \rho_{01}\\
        \rho_{10} & \rho_{11}
    \end{pmatrix} = \begin{pmatrix}
        \rho_{00} - \frac{\gamma}{1-\gamma}\rho_{11} & \frac{1}{\sqrt{1-\gamma}}\rho_{01}\\
        \frac{1}{\sqrt{1-\gamma}}\rho_{10} & \frac{1}{1-\gamma}\rho_{11}
    \end{pmatrix}, \\
    & \mc A_{\gamma_2}\circ \mc A_{\gamma_1}^{-1}\begin{pmatrix}
        \rho_{00} & \rho_{01}\\
        \rho_{10} & \rho_{11}
    \end{pmatrix} = \begin{pmatrix}
        \rho_{00} + \frac{\gamma_2 - \gamma_1}{1-\gamma_1}\rho_{11} & \frac{\sqrt{1-\gamma_2}}{\sqrt{1-\gamma_1}}\rho_{01}\\
        \frac{\sqrt{1-\gamma_2}}{\sqrt{1-\gamma_1}}\rho_{10} & \frac{1-\gamma_2}{1-\gamma_1}\rho_{11}
    \end{pmatrix}. 
\end{aligned}
\end{equation}
In particular, there exists a CPTP map $\mc D$ such that $\mc D \circ \mc A_{\gamma_1} = \mc A_{\gamma_2}$ if and only if $\mc A_{\gamma_2}\circ \mc A_{\gamma_1}^{-1}$ is CPTP if and only if $\gamma_1 \le \gamma_2$. If $\gamma_1 > \gamma_2$, $\mc A_{\gamma_2}\circ \mc A_{\gamma_1}^{-1}$ is non-positive, i.e., there exists $\sigma_0 \ge 0$ such that $\mc A_{\gamma_2}\circ \mc A_{\gamma_1}^{-1}(\sigma_0)$ has a negative eigenvalue. 

\noindent \textbf{Sufficiency}. We prove degradability by explicitly constructing the degrading map depicted in Figure \ref{fig:flagged degradable}. \\
\textbf{Case 1: $p=\frac{1}{2}$}. Since $\gamma+ \eta \le 1$, we have $\gamma \le 1- \eta$ and $\eta \le 1- \gamma$ then using \eqref{eqn:inversion of AD}, there exist qubit degrading maps $\mc D_1, \mc D_2$ such that 
\begin{equation}
    \mc D_1 \circ \mc A_{\gamma} = \mc A_{1-\eta},\quad \mc D_2 \circ \mc A_{\eta} = \mc A_{1-\gamma}.
\end{equation}
Then using the explicit formula of $(\Phi_{p,\gamma,\eta},  \Phi_{p,\gamma,\eta}^c)$ in \eqref{eqn:complementary pair of flagged mixture}, it is immediate to see that the degrading map $\mc D: \mb B(\mb C^2 \otimes \mb C^2) \to \mb B(\mb C^2 \otimes \mb C^2)$ can be chosen as
\begin{equation}
\begin{aligned}
        \mc D = \mc S_{10} \otimes \mc D_1 + \mc S_{01} \otimes \mc D_2,
\end{aligned}
\end{equation}
where the switch channel $\mc S_{ij}$ is defined by 
\begin{equation}\label{def:switch channel}
    \mc S_{ij}(\rho):= \bra{j}\rho \ket{j} \ketbra{i},\ i,j = 0,1.
\end{equation}
Note that the operational meaning of $\mc S_{ij}$ is to replace the flag $\ket{j}$ by $\ket{i}$. It is straightforward to check $\mc D \circ \Phi_{p,\gamma,\eta} = \Phi_{p,\gamma,\eta}^c$.

\noindent \textbf{Case 2: $1>p>\frac{1}{2}$}. Since $\eta \le \frac{1}{2}$ and $\gamma+ \eta \le 1$, using \eqref{eqn:inversion of AD}, there exist qubit degrading maps $\mc D_1, \mc D_2, \mc D_3$ such that 
\begin{equation}
    \mc D_1 \circ \mc A_{\gamma} = \mc A_{1-\eta},\quad \mc D_2 \circ \mc A_{\eta} = \mc A_{1-\eta},\quad \mc D_3 \circ \mc A_{\eta} = \mc A_{1-\gamma}
\end{equation}
Then one can see that the degrading map defined by 
\begin{equation}
    \mc D = \mc S_{10} \otimes \mc D_1 + \frac{2p - 1}{p} \mc S_{11} \otimes \mc D_2 + \frac{1-p}{p}\mc S_{01} \otimes \mc D_3
\end{equation}
satisfies $\mc D \circ \Phi_{p,\gamma,\eta}(\rho) = \Phi_{p,\gamma,\eta}(\rho)^c$. 

\noindent \textbf{Case 3: $0<p<\frac{1}{2}$}. Since $\gamma \le \frac{1}{2}$ and $\gamma+ \eta \le 1$, using \eqref{eqn:inversion of AD}, there exist qubit degrading maps $\mc D_1, \mc D_2, \mc D_3$ such that 
\begin{equation}
    \mc D_1 \circ \mc A_{\gamma} = \mc A_{1-\eta},\quad \mc D_2 \circ \mc A_{\gamma} = \mc A_{1-\gamma},\quad \mc D_3 \circ \mc A_{\eta} = \mc A_{1-\gamma}
\end{equation}
Then one can see that the degrading map defined by 
\begin{equation}
    \mc D = \frac{p}{1-p}\mc S_{10} \otimes \mc D_1 + \frac{1-2p}{1-p} \mc S_{00} \otimes \mc D_2 + \mc S_{01} \otimes \mc D_3
\end{equation}
satisfies $\mc D \circ \Phi_{p,\gamma,\eta}(\rho) = \Phi_{p,\gamma,\eta}(\rho)^c$. 

\noindent \textbf{Sufficiency}. The proof follows from proof by contradiction. In fact, we conclude the proof by showing that $\Phi_{p,\gamma,\eta}$ is non-degradable if $(p,\gamma, \eta)$ satisfies one of the following conditions:
    \begin{enumerate}
        \item $\eta + \gamma > 1$. 
        \item $p> \frac{1}{2}$: $\eta + \gamma \le 1$ and $\eta > \frac{1}{2}$.
        \item $p<\frac{1}{2}$: $\eta + \gamma \le 1$ and $\gamma > \frac{1}{2}$.
    \end{enumerate}
\textbf{Case 1: $\eta + \gamma > 1$}. Suppose in this case there exists a CPTP degrading map $\mc D$ such that $\mc D \circ \Phi_{p,\gamma,\eta}(\rho) = \Phi_{p,\gamma,\eta}(\rho)^c$. Then using the Kraus representation of $\mc D = \sum_k E_k \cdot E_k^{\dagger}$ where $E_k = \begin{pmatrix}
    E_{00}^k & E_{01}^k \\
    E_{10}^k & E_{11}^k
\end{pmatrix}, E_{ij}^k \in \mb B(\mb C^2)$, we have 
\begin{align*}
   & \sum_k \begin{pmatrix}
    E_{00}^k & E_{01}^k \\
    E_{10}^k & E_{11}^k
\end{pmatrix} \begin{pmatrix}
    (1-p)\mc A_{\gamma} & 0 \\
    0 & p\mc A_{\eta}
\end{pmatrix}\begin{pmatrix}
    (E_{00}^k)^{\dagger} & (E_{10}^k)^{\dagger} \\
    (E_{01}^k)^{\dagger} & (E_{11}^k)^{\dagger}
\end{pmatrix} = \begin{pmatrix}
    (1-p)\mc A_{1-\gamma} & 0 \\
    0 & p\mc A_{1- \eta}
\end{pmatrix}.
\end{align*}
Simplifying the above equation, and denote $\mc D_{ij} = \sum_k E_{ij}^k \cdot (E_{ij}^k)^{\dagger}$ we get 
\begin{equation}\label{eqn:key property of flagged mixture}
    \begin{cases}
        (1-p)\mc D_{00} \circ \mc A_{\gamma} + p \mc D_{01}\circ \mc A_{\eta} = (1-p) \mc A_{1-\gamma}, \\
        (1-p)\mc D_{10} \circ \mc A_{\gamma} + p \mc D_{11}\circ \mc A_{\eta} = p \mc A_{1-\eta}. 
    \end{cases}
\end{equation}
Note that using $\sum_k E_k^{\dagger} E_k = I_4$, $\mc D_{ij}$ are completely positive and trace decreasing such that $\mc D_{00} + \mc D_{01}$ and $\mc D_{10} + \mc D_{11}$ are quantum channels. Using the fact that $\mc A_{1-\gamma}\circ \mc A_{\eta}^{-1}$ is non-positive via $\eta > 1-\gamma$ and similarly $\mc A_{1-\eta}\circ \mc A_{\gamma}^{-1}$ is non-positive, \eqref{eqn:key property of flagged mixture} is given as 
\begin{equation}
    \begin{cases}
        (1-p)\mc D_{00} \circ \mc A_{\gamma}\circ \mc A_{\eta}^{-1} + p \mc D_{01} = (1-p) \mc A_{1-\gamma}\circ \mc A_{\eta}^{-1}, \\
        (1-p)\mc D_{10} + p \mc D_{11}\circ \mc A_{\eta}\circ \mc A_{\gamma}^{-1}  = p \mc A_{1-\eta}\circ \mc A_{\gamma}^{-1}. 
    \end{cases}
\end{equation}
On the left hand side, either $\mc A_{\gamma}\circ \mc A_{\eta}^{-1}$ or $\mc A_{\eta}\circ \mc A_{\gamma}^{-1}$ is completely positive but the right hand side is non-positive and we get a contradiction. Therefore, $\Phi_{p,\gamma,\eta}(\rho)$ is non-degradable.

\noindent \textbf{Case 2: $p> \frac{1}{2}$, $\eta + \gamma \le 1$ and $\eta > \frac{1}{2}$}. In this case, following the same calculation as the previous case, we arrive at \eqref{eqn:key property of flagged mixture}. We conclude the proof by showing that 
\begin{align}\label{eqn:non-degradable case 2}
    (1-p)\mc D_{10} \circ \mc A_{\gamma} + p \mc D_{11}\circ \mc A_{\eta} = p \mc A_{1-\eta}
\end{align}
is not possible, where $\mc D_{10}$ and $\mc D_{11}$ are completely positive and $\mc D_{10} + \mc D_{11}$ is trace-preserving. Denote 
\begin{align*}
    \mc T_{\mc D_{10}} = \begin{pmatrix}
        d_1 & * & * & d_3 \\
        \cdots & & & \cdots \\
        \cdots & & & \cdots \\
        d_2& * & * & d_4
    \end{pmatrix},\quad \mc T_{\mc D_{11}} = \begin{pmatrix}
        \wt d_1 & * & * & \wt d_3 \\
        \cdots & & & \cdots \\
        \cdots & & & \cdots \\
        \wt d_2& * & * & \wt d_4
    \end{pmatrix}.
\end{align*}
Using the relation between Choi–Jamio\l{}kowski opertor and transfer matrix \eqref{eqn:relation Choi and transfer}, the Choi–Jamio\l{}kowski opertors $\mc J_{\mc D_{10}}, \mc J_{\mc D_{11}} \ge 0$ are given by 
\begin{align*}
    \mc J_{\mc D_{10}} = \begin{pmatrix}
        d_1 & * & * & * \\
        * & d_2& * & * \\
        * & * & d_3 & * \\
        * & * & * & d_4
    \end{pmatrix},\  \mc J_{\mc D_{11}} = \begin{pmatrix}
        \wt d_1 & * & * & * \\
        * & \wt d_2 & * & * \\
        * & * & \wt d_3 & * \\
        * & * & * & \wt d_4
    \end{pmatrix}.
\end{align*}
Also note that $\tr_2 (\mc J_{\mc D_{10}}+ \mc J_{\mc D_{11}}) = I_2$, the restriction on $d_i,\wt d_i$ is given by
\begin{equation}
\begin{aligned}
        &d_i, \wt d_i \ge 0,\quad 1\le i \le 4;\\
        &d_1 + \wt d_1 + d_2 + \wt d_2 = 1; \ d_3 + \wt d_3 + d_4 + \wt d_4 = 1.
\end{aligned}
\end{equation}
Using the composition rule for transfer matrix \eqref{eqn:composition of transfer matrix}, we have $(1-p)\mc T_{\mc D_{10}} \mc T_{\mc A_{\gamma}} + p \mc T_{\mc D_{11}} \mc T_{\mc A_{\eta}} = p \mc T_{\mc A_{1-\eta}}$, where the transfer matrix of amplitude damping channel is 
\begin{align*}
    \mc T_{\mc A_{\gamma}} = \begin{pmatrix}
        1 & 0 & 0 & \gamma \\
        0 & \sqrt{1-\gamma} & 0 & 0 \\
        0 & 0 & \sqrt{1-\gamma} & 0 \\
        0 & 0 & 0 & 1- \gamma
    \end{pmatrix},
\end{align*}
and compare the four corner elements, we have 
\begin{align*}
    \begin{cases}
        (1-p)d_1 + p \wt d_1 = p,\\
        (1-p)d_2 + p \wt d_2 = 0, \\
        (1-p)[d_1 \gamma + d_3(1-\gamma)] + p[\wt d_1 \eta + \wt d_3(1-\eta)] = p(1-\eta), \\
        (1-p)[d_2 \gamma + d_4(1-\gamma)] + p[\wt d_2 \eta + \wt d_4(1-\eta)] = p\eta. \\
    \end{cases}
\end{align*}
Using the positivity of $d_i, \wt d_i$, we can conclude by elementary algebra that the only possible solution is $d_i = 0, 1\le i\le 4$ thus $\mc D_{10} = 0$. In fact, it is easy to see that $d_1 = d_2 = \wt d_2 = 0,\ \wt d_1 = 1$. Then the last two equations simplify as 
\begin{align*}
    & (1-p)(1-\gamma)d_3 + p(1-\eta) \wt d_3 = p(1-2\eta), \\
    & (1-p)(1-\gamma)d_4 + p(1-\eta) \wt d_4 = p\eta.
\end{align*}
Taking the sum, we see that the only possible solution is $d_3 + d_4 = 0$ thus $\mc D_{10} = 0$. The equation \eqref{eqn:non-degradable case 2} becomes $\mc D_{11} \circ \mc A_{\eta} = \mc A_{1-\eta}$, which is not possible because $\eta > \frac{1}{2}$.

\noindent \textbf{Case 3: $p< \frac{1}{2}$, $\eta + \gamma \le 1$ and $\gamma > \frac{1}{2}$}. This case follows from the same argument as Case 2. In fact, using \eqref{eqn:key property of flagged mixture}, we can conclude the proof by showing that 
\begin{align*}
    (1-p)\mc D_{00} \circ \mc A_{\gamma} + p \mc D_{01}\circ \mc A_{\eta} = (1-p) \mc A_{1-\gamma}
\end{align*}
is not possible. Then the same calculation in Case 2 holds if we replace $\eta$ by $\gamma$ and $p$ by $1-p$. 
\end{proof}

\subsection{Informationally degradable and anti-degradable regions}
In the previous subsection, we characterize the regions where the channel $\Phi_{p,\gamma,\eta}$ is neither degradable nor anti-degradable. In this subsection, we provide evidence (rigorous proof for special cases and numerical evidence in full generality) that for any $(\gamma,\eta)$ in that region, there exists a threshold $p^*(\gamma,\eta)$ given by Proposition \ref{main:less noisy}, such that when $p$ is above or below the threshold, we have informational degradability or informational anti-degradability introduced in Definition \ref{def:weaker degradability}. \added{To be more specific, we establish the positivity of $R_3$, defined in \eqref{ratio:3}. }

There are four separate regions, see Figure \ref{fig:deg region} and we only consider the region $$p>\frac{1}{2},\ \gamma >\frac{1}{2},\ \eta < \frac{1}{2},\ \gamma + \eta >1,$$ since the other regions are similar. 

Recall that Proposition \ref{main:less noisy} shows that the channel is informational degradable for $p$ greater than a threshold, if $R_3$ defined in \eqref{ratio:3} is positive:
\begin{align*}
    R_3(\mc A_{\eta}^{A \to B_1}, \mc A_{1-\gamma}^{A \to B_2}) = \inf_{\mc H_V} \inf_{\rho_{VA}} \frac{I(V;B_1) - I(V;E_1)}{I(V;B_2) - I(V;E_2)} > 0.
\end{align*}

\subsubsection{Positivity of $R_3$---a necessary condition}
To ensure that $R_3$ is strictly positive, it is necessary to exclude the scenario where $I(V;B_1) = I(V;E_1)$ but $I(V;B_2) > I(V;E_2)$ for some state $\rho_{VA}$. The following proposition establishes that $I(V;B_i) = I(V;E_i)$ if and only if $\rho_{VA} = \rho_V \otimes \rho_A$.
\begin{prop}\label{main:uniqueness of fixed points}
    Suppose $\mc N = \mc N^{B\to E}$ has a unique fixed state, i.e., there exists a unique quantum state $\rho_0$ such that $\mc N(\rho_0) = \rho_0$. Then for any finite-dimensional quantum system $\mc H_V$ and quantum state $\rho_{VB}$, 
    \begin{equation}\label{fixed point}
        \left(id_{\mb B(\mc H_V)} \otimes \mc N\right) (\rho_{VB}) = \rho_{VB}
    \end{equation}
    if and only if $\rho_{VB} = \rho_V \otimes \rho_B$ and $\rho_B = \rho_0$. 
\end{prop}
The proof of Proposition \ref{main:uniqueness of fixed points} is provided in Appendix \ref{appendix:equality case}. For a degradable amplitude damping channel $\mc A_{\gamma}$ ($\gamma \leq \frac{1}{2}$) with output $B$ and environment $E$, we can verify that $I(V;B) = I(V;E)$ if and only if $\rho_{VA} = \rho_V \otimes \rho_A$ by calculating
\begin{align*}
    & I(V;B) = D(\rho_{VB}\| \rho_V \otimes \rho_B),\ I(V;E) = D(\rho_{VE}\| \rho_V \otimes \rho_E),
\end{align*}
where 
\begin{align*}
    \rho_{VE} = \big(id_{\mb B(\mc H_V)} \otimes \mc A_{\frac{1-2\gamma}{1-\gamma}}\big)(\rho_{VB}).
\end{align*}
By the data processing inequality, we have $I(V;B) \geq I(V;E)$, with equality ($I(V;B) = I(V;E)$) if and only if \cite{uhlmann1977relative, nielsen2004simple, petz1986quasi}
\begin{equation}\label{recovery channel}
    \mc R_{id_{\mb B(\mc H_V)}\otimes \mc A_{\frac{1-2\gamma}{1-\gamma}}, \rho_V\otimes \rho_B}
    \big((id_{\mb B(\mc H_V)}\otimes \mc A_{\frac{1-2\gamma}{1-\gamma}})(\rho_{VB})\big) = \rho_{VB},
\end{equation}
where the recovery map $\mc R_{\mc A, \sigma}$ for a given quantum channel $\mc A$ and quantum state $\sigma$ is defined as
\begin{equation}
    \mc R_{\mc A, \sigma}(X) := \sigma^{1/2} \mc A^*\big(\mc A(\sigma)^{-1/2} X \mc A(\sigma)^{-1/2} \big)\sigma^{1/2},
\end{equation}
where $\mc A^{*}$ is the dual map of $\mc A$, defined by $\Tr(\mc A(\rho)X) = \Tr(\rho \mc A^*(X)), \forall \rho, X$. Here the channel is only defined when $\text{supp}(X) \subseteq \text{supp}(\mc A(\sigma))$. 

It is straightforward to verify that the composition of the recovery map and the amplitude damping channel results in a channel with the form
\begin{equation}\label{eqn:recove amplitude damping}
    id_{\mb B(\mc H_V)}\otimes \mc N,
\end{equation}
where $\mc N$ has a unique fixed point; the details are presented in Appendix \ref{appendix:equality case}. Consequently, Proposition \ref{main:uniqueness of fixed points} implies that for a degradable amplitude damping channel, $I(V;B) = I(V;E)$ if and only if $\rho_{VA} = \rho_V \otimes \rho_A$.

\subsubsection{Complete version of contraction and expansion coefficient}
\added{We propose a framework that offers a sufficient condition ensuring the positivity of \(R_3\).} Let \(\mc N\) and \(\mc M\) be two quantum channels with isometries given by 
\begin{align*}
    U_{\mc N}: \mc H_A \to \mc H_{B_1} \otimes \mc H_{E_1}, \quad U_{\mc M}: \mc H_A \to \mc H_{B_2} \otimes \mc H_{E_2}.
\end{align*}
We define the \textit{complete} contraction and expansion coefficients of \((\mc N, \mc M)\) as follows:
\begin{align}\label{eqn:complete contraction and expansion}
    \eta_{\mc N,\mc M}^{cb} &= \sup_{\rho_{VA}} \frac{I(V;B_1)}{I(V;B_2)}, \quad 
    \widecheck{\eta}_{\mc N,\mc M}^{cb} = \inf_{\rho_{VA}} \frac{I(V;B_1)}{I(V;B_2)},
\end{align}
where the optimization is performed over all valid states \(\rho_{VA}\) on \(\mc H_V \otimes \mc H_A\).

\added{The following conjecture is a sufficient condition for positivity of $R_3$}:
\begin{conj}\label{conjecture}
    Suppose \(0 < \gamma_2 < \gamma_1 < 1\). Then, for amplitude damping channels \(\mc A_{\gamma_1}\) and \(\mc A_{\gamma_2}\), the following holds:
    \begin{equation}
        \eta_{\mc A_{\gamma_1},\mc A_{\gamma_2}}^{cb} < 1, \quad \widecheck{\eta}_{\mc A_{\gamma_1},\mc A_{\gamma_2}}^{cb} > 0.
    \end{equation}
\end{conj}
Note that \(\widecheck{\eta}_{\mc A_{\gamma_1},\mc A_{\gamma_2}}^{cb} > 0\) can be interpreted as a complete version of the \textit{reverse-type data processing inequality}. Specifically, since there exists a quantum channel \(\mc D\) such that \(\mc A_{\gamma_1} = \mc D \circ \mc A_{\gamma_2}\) given by \eqref{eqn:inversion of AD}, we are seeking a universal constant \(c > 0\) such that 
\begin{align*}
    & D\left( \mc D \circ \mc A_{\gamma_2}(\rho_{VA})\| \mc D \circ \mc A_{\gamma_2}(\rho_V \otimes \rho_A)\right) \geq c \, D\left(\mc A_{\gamma_2}(\rho_{VA})\| \mc A_{\gamma_2}(\rho_V \otimes \rho_A)\right).
\end{align*}
If the quantum state \(\rho_{VA}\) is restricted to classical-quantum states, the conjecture has been proven in \cite[Proposition 5.8, Lemma 6.3]{BGSW24}. However, for the fully quantum case, the entanglement between the ancillary system and the input system introduces significant challenges. Another possible route when $\gamma_1,\gamma_2$ are close is to use stability trick studied in \cite{cubitt2015stability, junge2023stability, junge2022stability}.

Here, we present a special case to illustrate why \(\widecheck{\eta}_{\mc A_{\gamma_1},\mc A_{\gamma_2}}^{cb} > 0\) can still hold in the presence of entanglement. By Proposition \ref{main:uniqueness of fixed points}, we know that \(I(V;B)\) approaches zero when \(\rho_{VA}\) is close to a product state due to the continuity of the relative entropy. To explore the infinitesimal behavior, we consider the following entangled operators parametrized by \(\varepsilon\), which are constructed by perturbing a product state with an entangled state $\rho_{VA}(\varepsilon) = \rho_V \otimes \rho_A + \varepsilon \sigma_{VA}$, where \(\sigma_{VA}\) is a traceless Hermitian operator encoding the entanglement, and \(\varepsilon \ll 1\). Now we choose a special product state and entangled operator, resulting in the following construction:
\begin{equation}
\begin{aligned}
    & \rho_{VA}(\varepsilon)= (1-\varepsilon)\ketbra{1}{1}_V \otimes \ketbra{1}{1}_A + \frac{\varepsilon}{2} \left(\ketbra{01}{01}_{VA}  + \ketbra{10}{10}_{VA} + \ketbra{01}{10}_{VA} + \ketbra{10}{01}_{VA}\right) 
    = \begin{pmatrix}
        0 & 0 & 0 & 0\\
        0 & \frac{\varepsilon}{2} & \frac{\varepsilon}{2} & 0 \\
        0 &  \frac{\varepsilon}{2} & \frac{\varepsilon}{2} & 0 \\
        0 & 0 & 0 & 1-\varepsilon
    \end{pmatrix}.
\end{aligned}
\end{equation}
Then, denote $\rho_{VB}(\varepsilon) = (id_{\mb B(\mc H_V)} \otimes \mc A_\gamma)(\rho_{VA}(\varepsilon))$, it is straightforward to calculate that
\begin{align*}
    &\rho_{VB}(\varepsilon)  = (1-\varepsilon)\ketbra{1}{1}_V \otimes (\gamma \ketbra{0}{0}_A + (1-\gamma)\ketbra{1}{1}_A )  + \frac{\varepsilon}{2} \left(\ketbra{0}{0}_V \otimes (\gamma \ketbra{0}{0}_A + (1-\gamma)\ketbra{1}{1}_A ) + \ketbra{10}{10}_{VA} \right) \\
    & + \frac{\varepsilon}{2} \sqrt{1-\gamma}\left(\ketbra{0}{1}_V \otimes \ketbra{1}{0}_A + \ketbra{1}{0}_V \otimes \ketbra{0}{1}_A \right) = \begin{pmatrix}
        \frac{\varepsilon}{2} \gamma & 0 & 0 & 0\\
        0 & \frac{\varepsilon}{2}(1-\gamma) & \frac{\varepsilon}{2}\sqrt{1-\gamma} & 0 \\
        0 &  \frac{\varepsilon}{2}\sqrt{1-\gamma} & \frac{\varepsilon}{2} + (1-\varepsilon)\gamma & 0 \\
        0 & 0 & 0 & (1-\varepsilon)(1-\gamma)
    \end{pmatrix}.
\end{align*}
Using the formula for eigenvalues of two by two matrices, and $(1+x)^{\frac{1}{2}} = 1 + \frac{1}{2}x - \frac{1}{8}x^2 + \mc O(x^3)$, we can expand the eigenvalues to second order in \(\varepsilon\):
\begin{align*}
   & \lambda_{1}(\varepsilon) = \frac{\varepsilon}{2}\gamma,
 \\
& \lambda_{2}(\varepsilon) = \frac{\gamma - \frac{3\gamma}{2} \varepsilon + \varepsilon + \gamma \sqrt{1 - \varepsilon + \frac{(\gamma-2)^2}{4\gamma^2} \varepsilon^2}}{2} = \gamma + \frac{1-2\gamma}{2}\varepsilon + \frac{1-\gamma}{4\gamma}\varepsilon^2 + \mc O(\varepsilon^3), \\
& \lambda_{3}(\varepsilon) = \frac{\gamma - \frac{3\gamma}{2} \varepsilon + \varepsilon - \gamma \sqrt{1 - \varepsilon + \frac{(\gamma-2)^2}{4\gamma^2} \varepsilon^2}}{2} = \frac{1-\gamma}{2}\varepsilon - \frac{1-\gamma}{4\gamma}\varepsilon^2 + \mc O(\varepsilon^3),  \\
& \lambda_{4}(\varepsilon)= (1-\varepsilon)(1-\gamma).
\end{align*}
The reduced states are given by
\begin{align*}
    &\rho_V(\varepsilon) =  \frac{\varepsilon}{2} \ketbra{0}{0} +  (1-\frac{\varepsilon}{2})\ketbra{1}{1}, \\
    &\rho_B(\varepsilon) = (\frac{\varepsilon}{2} + \gamma(1-\frac{\varepsilon}{2})) \ketbra{0}{0} + (1-\frac{\varepsilon}{2})(1-\gamma) \ketbra{1}{1}.
\end{align*}
Using the fact 
\begin{equation}\label{eq:1}
  \begin{aligned}
    & (a+\delta) \log (a+\delta)= a \log a + (\log a + 1)\delta + \frac{1}{2a} \delta^2 + O(\delta^3),\quad a>0,\ \delta \to 0,
\end{aligned}  
\end{equation}
the mutual information is calculated as follows:
\begin{align*}
    & I(V;B) = S(V) + S(B) - S(BV) = h(\frac{\varepsilon}{2})+ h((1-\frac{\varepsilon}{2})(1-\gamma)) + \sum_{i=1}^4 \lambda_i(\varepsilon) \log \lambda_i(\varepsilon)\\
    & = -\frac{\varepsilon}{2}\log \frac{\varepsilon}{2} - (1-\frac{\varepsilon}{2})\log(1-\frac{\varepsilon}{2}) - \big((1-\frac{\varepsilon}{2})(1-\gamma)\big)\log \big((1-\frac{\varepsilon}{2})(1-\gamma)\big)  - \big(\frac{\varepsilon}{2} + \gamma(1-\frac{\varepsilon}{2})\big)\log \big(\frac{\varepsilon}{2} + \gamma(1-\frac{\varepsilon}{2})\big)  \\
    & + \frac{\varepsilon}{2}\gamma \log (\frac{\varepsilon}{2}\gamma)  + \big(\gamma + \frac{1-2\gamma}{2}\varepsilon + \frac{1-\gamma}{4\gamma}\varepsilon^2 \big)\log\big(\gamma + \frac{1-2\gamma}{2}\varepsilon + \frac{1-\gamma}{4\gamma}\varepsilon^2 \big)  + \big(\frac{1-\gamma}{2}\varepsilon - \frac{1-\gamma}{4\gamma}\varepsilon^2 \big) \log \big(\frac{1-\gamma}{2}\varepsilon - \frac{1-\gamma}{4\gamma}\varepsilon^2 \big) \\
    & + \big((1-\varepsilon)(1-\gamma)\big) \log \big((1-\varepsilon)(1-\gamma)\big) + \mc O(\varepsilon^3) =: C_0 + C_1 \varepsilon + C_2'\varepsilon^2 \log \varepsilon + C_2 \varepsilon^2 + \mc O(\varepsilon^3).
\end{align*}
Using \eqref{eq:1}, we can show that $C_0 = C_1 = 0$. Thus the asymtotic behavior is determined by $C_2'$. Note that $C_2'$ originates from 
\begin{align*}
    \big(\frac{1-\gamma}{2}\varepsilon - \frac{1-\gamma}{4\gamma}\varepsilon^2 \big) \log \big(\frac{1-\gamma}{2}\varepsilon - \frac{1-\gamma}{4\gamma}\varepsilon^2 \big)
\end{align*}
and is given by $-\frac{1-\gamma}{4\gamma}$. Thus for this choice of entangled states, 
\begin{equation}
   I(V;B) \sim -\frac{1-\gamma}{4\gamma} \varepsilon^2 \log \varepsilon.
\end{equation}
Therefore, for two parameters \(0 < \gamma_2 < \gamma_1 < 1\), the ratio of mutual information is lower bounded using the special family of entangled states that are close to the product state \(\ketbra{11}{11}_{VA}\). Numerical results strongly suggest that the lower bound is achieved near this product state, leading to the conjectured closed formula:
\begin{equation}
    \widecheck{\eta}_{\mc A_{\gamma_1},\mc A_{\gamma_2}}^{cb} = \inf_{\rho_{VA}} \frac{I(V;B_1)}{I(V;B_2)} = \frac{\gamma_2(1-\gamma_1)}{\gamma_1(1-\gamma_2)}.
\end{equation}
We cannot rigorously prove this conjecture in full generality, and we leave it as an open question.

Finally, assuming Conjecture \ref{conjecture} holds, we can show that \(R_3(\mc A_{\eta}^{A \to B_1}, \mc A_{1-\gamma}^{A \to B_2}) > 0\) when \(\gamma > \frac{1}{2},\ \eta < \frac{1}{2},\ \gamma + \eta > 1\). Specifically, we have
\begin{align*}
    \frac{I(V;B_1) - I(V;E_1)}{I(V;B_2) - I(V;E_2)} 
    &= \frac{I(V;B_1)}{I(V;B_2)} \left(\frac{1 - \frac{I(V;E_1)}{I(V;B_1)}}{1 - \frac{I(V;E_2)}{I(V;B_2)}} \right) \geq \frac{I(V;B_1)}{I(V;B_2)} \left(1 - \frac{I(V;E_1)}{I(V;B_1)}\right) \geq \widecheck{\eta}_{\mc A_{\eta},\mc A_{1-\gamma}}^{cb} \big(1 - \eta_{\mc A_{1-\eta},\mc A_{\eta}}^{cb}\big) > 0.
\end{align*}

Therefore, for any \(\gamma, \eta\) satisfying \(\gamma > \frac{1}{2},\ \eta < \frac{1}{2},\ \gamma + \eta > 1\), and using Proposition \ref{main:less noisy} and Proposition \ref{main:degradable characterization}, the quantum channel \(\Phi_{p,\gamma,\eta}\) is informationally degradable for \(p \geq \frac{1}{1 + R_3(\mc A_{\eta}, \mc A_{1-\gamma})}\), but it is neither degradable nor anti-degradable.

\begin{remark}
    In \cite{sutter2017approximate,fanizza2020quantum, kianvash2022bounding, wang2021pursuing}, an upper bound on the quantum capacity based on approximate degradability and flagged extension is given. Here our example shows that approximate degradability for flagged channels will not always give us a good bound on the capacity, since additivity may still hold even if it is non-degradable.
\end{remark}
\section{Additivity and non-additivity properties for generalized Platypus channels}\label{sec:Platypus}
In this section, we introduce a two-parameter quantum channel \(\mc N_{s,t}\) with \(0 \leq s, t \leq 1\) and \(s + t \leq 1\), which generalizes the Platypus channels introduced in \cite{leditzky2023platypus}. This class of channels exhibits rich additivity and non-additivity properties. In the parameter region where the channel is neither degradable nor anti-degradable, we show the following:
\begin{itemize}
    \item \textbf{Weak and strong additivity for special cases}:  
    If \(s = 0\) or \(s + t = 1\), the channels exhibit weak additivity of coherent information and strong additivity of coherent information when paired with degradable channels. This result follows from Theorem \ref{thm:additivity via DPI}.

    \item \textbf{Failure of strong additivity}:  
    If \(s > 0\) and \(s + t < 1\), strong additivity with degradable channels fails. In fact, within this parameter range, we observe phenomena such as super-activation and amplification effects when tensoring with erasure channels (and some other degradable channels). Depending on the specific parameter region, the arguments for strong non-additivity vary. These include the Smith-Yard argument introduced in \cite{smith2008quantum} and the log-singularity argument developed in \cite{siddhu2021entropic}.
\end{itemize}
\added{A summary of the main results in this section is presented as follows:}
\begin{figure}[H]
    \centering\includegraphics[width=1\textwidth]{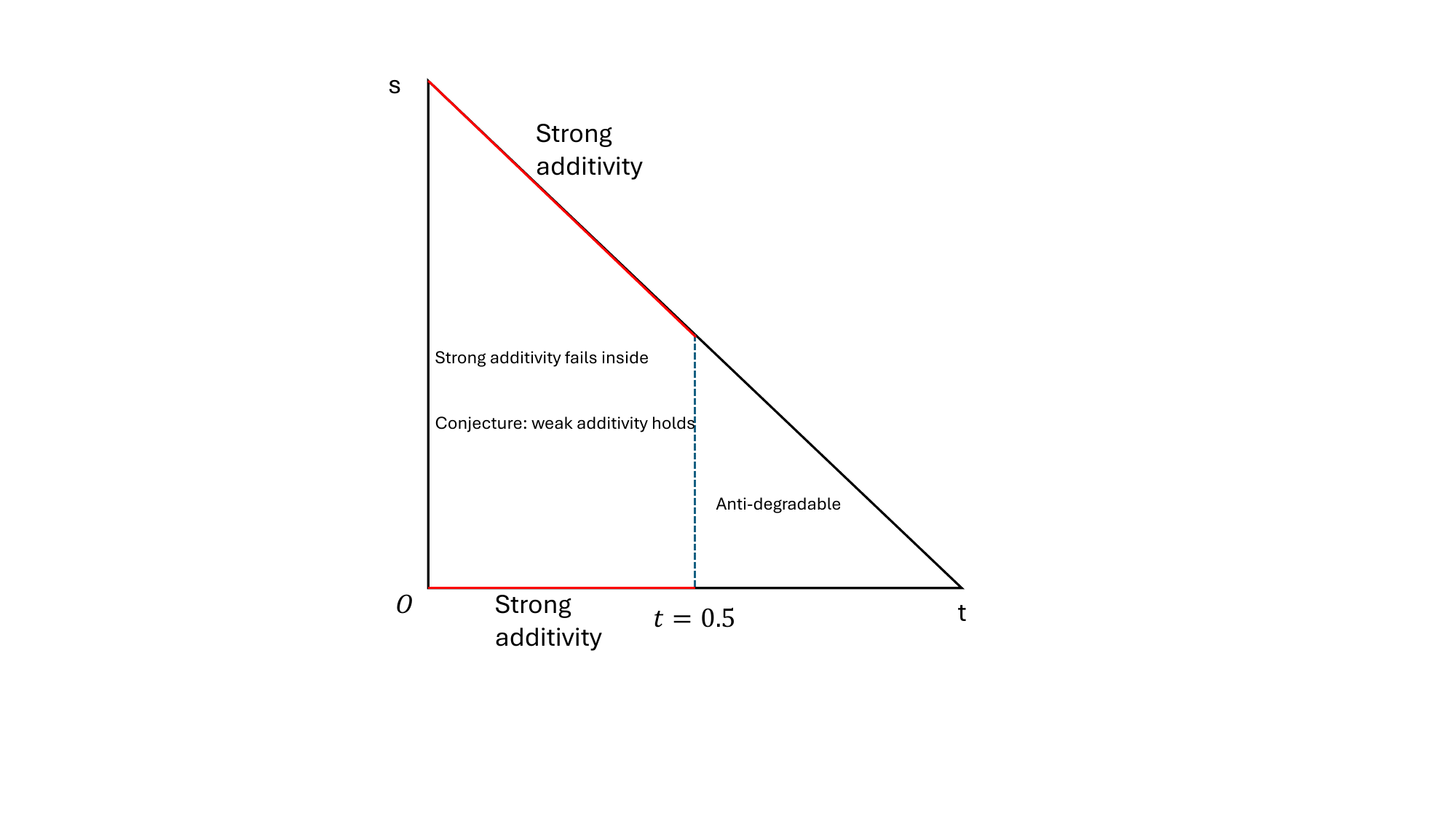}
    \caption{Additivity and non-additivity properties for $\mc N_{s,t}$.}
\label{figure:summary Platypus}
\end{figure}

\subsection{Basic properties of generalized Platypus channels}
Consider an isometry $F_{s,t}: \mc H_A \to \mc H_B \otimes \mc H_E$ with $\rm{dim}\mc H_A = \rm{dim}\mc H_B = \rm{dim}\mc H_E = 3$ of the form:
\begin{equation}\label{eqn:Generalized Platypus}
    \begin{aligned}
        & F_{s,t} \ket{0}  = \sqrt{s} \ket{0}\otimes \ket{0} + \sqrt{1-s-t} \ket{1}\otimes \ket{1}+ \sqrt{t} \ket{2} \otimes \ket{2}, \\
        & F_{s,t} \ket{1} = \ket{2} \otimes \ket{0}, \\
        & F_{s,t} \ket{2} = \ket{2} \otimes \ket{1},
    \end{aligned}
\end{equation}
where $0\le s,t \le 1$ with $s+t \le 1$. We denote the complementary pair as $(\mc N_{s,t}, \mc N_{s,t}^c)$ with $\mc N_{s,t}(\rho) := \Tr_{E}(F_{s,t} \rho F_{s,t}^{\dagger}),\ \mc N^c_{s,t}(\rho) := \Tr_B(F_{s,t} \rho F_{s,t}^{\dagger})$. In the matrix form, for $\rho = \sum_{i,j=0}^2 \rho_{ij} \ketbra{i}{j}$, we have 
\begin{equation}\label{eqn:matrix form of platypus}
\begin{aligned}
    & \mc N_{s,t}(\rho) = \begin{pmatrix}
        s \rho_{00} & 0 & \sqrt{s} \rho_{01} \\
        0 & (1-s-t)\rho_{00} & \sqrt{1-s-t} \rho_{02} \\
        \sqrt{s} \rho_{10} & \sqrt{1-s-t}\rho_{20} & t\rho_{00} + \rho_{11} + \rho_{22}
    \end{pmatrix},   \\
    & \mc N^c_{s,t}(\rho) = \begin{pmatrix}
        s \rho_{00} + \rho_{11} & \rho_{12} & \sqrt{t} \rho_{10} \\
        \rho_{21} & (1-s-t)\rho_{00}+ \rho_{22} & \sqrt{t} \rho_{20} \\
        \sqrt{t} \rho_{01} & \sqrt{t}\rho_{02} & t\rho_{00} 
    \end{pmatrix}.
\end{aligned}
\end{equation}
In terms of Kraus representation, we have $\mc N_{s,t}(\rho) = \sum_{k=0}^2 E_k \rho E_k^{\dagger},\ \mc N_{s,t}^c(\rho) = \sum_{k=0}^2 \wt E_k \rho  \wt E_k^{\dagger}$, where 
\begin{align*}
    &E_0 = \sqrt{s} \ketbra{0}{0} + \ketbra{2}{1},\ E_1 = \sqrt{1-s-t} \ketbra{1}{0} + \ketbra{2}{2},\ E_2 = \sqrt{t} \ketbra{2}{0},
\end{align*}
and 
\begin{align*}
    &\wt E_0 = \sqrt{s} \ketbra{0}{0},\ \wt E_1 = \sqrt{1-s-t} \ketbra{1}{0},\ \wt E_2 = \sqrt{t} \ketbra{2}{0} + \ketbra{0}{1} + \ketbra{1}{2}.
\end{align*}
In terms of transfer matrix, arranging the order of basis as $\{\ket{00}, \ket{01}, \ket{02}, \ket{10}, \ket{11}, \ket{12}, \ket{20}, \ket{21}, \ket{22}\}$, we have 
\begin{align*}
    \mc T_{\mc N_{s,t}} = \begin{pmatrix}
            s & 0 & 0 & 0 & 0 & 0 & 0 & 0 & 0  \\
            0 & 0 & 0 & 0 & 0 & 0 & 0 & 0 & 0  \\
            0 & \sqrt{s} & 0 & 0 & 0 & 0 & 0 & 0 & 0  \\
            0 & 0 & 0 & 0 & 0 & 0 & 0 & 0 & 0  \\
            u_{s,t} & 0 & 0 & 0 & 0 & 0 & 0 & 0 & 0  \\
            0 & 0 & \sqrt{u_{s,t}} & 0 & 0 & 0 & 0 & 0 & 0  \\
            0 & 0 & 0 & \sqrt{s} & 0 & 0 & 0 & 0 & 0  \\
            0 & 0 & 0 & 0 & 0 & 0 & \sqrt{u_{s,t}} & 0 & 0  \\
            t & 0 & 0 & 0 & 1 & 0 & 0 & 0 & 1  \\
        \end{pmatrix},\ 
    \mc T_{\mc N^c_{s,t}} = \begin{pmatrix}
            s & 0 & 0 & 0 & 1 & 0 & 0 & 0 & 0  \\
            0 & 0 & 0 & 0 & 0 & 1 & 0 & 0 & 0  \\
            0 & 0 & 0 & \sqrt{t} & 0 & 0 & 0 & 0 & 0  \\
            0 & 0 & 0 & 0 & 0 & 0 & 0 & 1 & 0  \\
            u_{s,t} & 0 & 0 & 0 & 0 & 0 & 0 & 0 & 1  \\
            0 & 0 & 0 & 0 & 0 & 0 & \sqrt{t} & 0 & 0  \\
            0 & \sqrt{t} & 0 & 0 & 0 & 0 & 0 & 0 & 0  \\
            0 & 0 & \sqrt{t} & 0 & 0 & 0 & 0 & 0 & 0  \\
            t & 0 & 0 & 0 & 0 & 0 & 0 & 0 & 0  \\
        \end{pmatrix}.
\end{align*}
Here we denote $u_{s,t} = 1-s-t$. It is straightforward to see that there is no matrix $D$ such that $D \mc T_{\mc N_{s,t}} = \mc T_{\mc N^c_{s,t}}$, since the sixth and eighth column of $\mc T_{\mc N_{s,t}}$ is zero but $\mc T_{\mc N^c_{s,t}}$ has non-zero sixth and eighth column. Therefore, for any $s,t$ the channel is not degradable via composition rule for transfer matrix. 

To see antidegrability, assume there is a superoperator $\mc D: \mb B(\mc H_E) \to \mb B(\mc H_B)$ such that $\mc D \circ \mc N^c_{s,t} = \mc N_{s,t}$, then using the composition rule of transfer matrix, we have $\mc T_{\mc D}\mc T_{\mc N^c_{s,t}} = \mc T_{\mc N_{s,t}}$. Moreover, when $t>0$, $\mc T_{\mc N^c_{s,t}}$ is invertible thus $\mc T_{\mc D} = \mc T_{\mc N_{s,t}} \mc T_{\mc N^c_{s,t}}^{-1}$ and we only need to determine whether it generates a CPTP map. It is straightforward to calculate $\mc T_{\mc N^c_{s,t}}^{-1}$. Then calculating matrix multiplication and using the relation between transfer matrix and Choi–Jamio\l{}kowski operator \eqref{eqn:relation Choi and transfer}, we have 
\begin{equation}\label{eqn:Transfer to Choi}
\begin{aligned}
    \mc J_{\mc D} = \begin{pmatrix}
            0 & 0 & 0 & 0 & 0 & 0 & 0 & 0 & 0  \\
            0 & 0 & 0 & 0 & 0 & 0 & 0 & 0 & 0  \\
            0 & 0 & 1 & 0 & 0 & 0 & \frac{\sqrt{s}}{\sqrt{t}} & 0 & 0  \\
            0 & 0 & 0 & 0 & 0 & 0 & 0 & 0 & 0  \\
            0 & 0 & 0 & 0 & 0 & 0 & 0 & 0 & 0  \\
            0 & 0 & 0 & 0 & 0 & 1 & 0 & \frac{\sqrt{u_{s,t}}}{\sqrt{t}} & 0  \\
            0 & 0 & \frac{\sqrt{s}}{\sqrt{t}} & 0 & 0 & 0 & \frac{s}{t} & 0 & 0  \\
            0 & 0 & 0 & 0 & 0 & \frac{\sqrt{u_{s,t}}}{\sqrt{t}} & 0 & \frac{u_{s,t}}{t} & 0  \\
            0 & 0 & 0 & 0 & 0 & 0 & 0 & 0 & \frac{2t-1}{t}  \\
        \end{pmatrix}.
\end{aligned}
\end{equation}
Note that $\mc D$ is a CPTP map if and only if $\mc J_{\mc D}$ is positive and $\Tr_2(\mc J_{\mc D}) = I$ if and only if $\frac{2t-1}{t}\ge 0\iff t\ge \frac{1}{2}$. In summary, we show the following:
\begin{prop}
    If $t\ge \frac{1}{2}$, $\mc N_{s,t}$ is anti-degradable; if $t < \frac{1}{2}$, $\mc N_{s,t}$ is neither degradable nor anti-degradable.
\end{prop}

\subsection{Subchannels and optimizer of coherent information}
For the quantum channel $\mc N_{s,t}$ where $0<t<1/2, s>0$ and $s+t < 1$, the coherent information given by
\begin{align*}
    \mc Q^{(1)}(\mc N_{s,t}) = \max_{\rho} S(\mc N_{s,t}(\rho)) - S(\mc N^c_{s,t}(\rho))
\end{align*}
is achieved at 
\begin{equation}\label{eqn:optimizer of generalized platypus}
    \begin{cases}
       u \ketbra{0}{0} + (1-u) \ketbra{2}{2},\quad 1-s-t \ge s \\
       u \ketbra{0}{0} + (1-u) \ketbra{1}{1},\quad 1-s-t \le s \\
    \end{cases}
\end{equation}
A proof of the above fact is given in Appendix \ref{app:optimizer}. Since the maximal coherent information is given on the subspace, we introduce the subchannel, or restriction channel, denoted as $\widehat{\mc N}_{s,t}$ defined by
\begin{equation}\label{def:restriction channel}
    \widehat{\mc N}_{s,t} = \begin{cases}
        \mc N_{s,t}\big|_{\text{span}\{\ket{0}, \ket{1}\}},\ s \ge \frac{1-t}{2}, \\
        \mc N_{s,t}\big|_{\text{span}\{\ket{0}, \ket{2}\}},\ s \le \frac{1-t}{2}.
    \end{cases}
\end{equation}
Recalling Definition~\ref{def:weak_dom}, one sees that \(\mc N_{s,t}\) is \emph{weakly dominated} by its subchannel \(\widehat{\mc N}_{s,t}\). This observation is especially helpful for analyzing the behavior of \(\mc N_{s,t}\) in the regime where it is neither degradable nor anti-degradable, since it allows us to focus on a lower-dimensional restriction that captures the essential behavior of the coherent information. The subchannel has the following property:
\begin{prop}
    $\widehat{\mc N}_{s,t}$ defined by \eqref{def:restriction channel} is either degradable or anti-degradable. 
\end{prop}
\begin{proof}
\textbf{Case I: $s \ge 1-s-t$.} In the matricial form, $ \widehat{\mc N}_{s,t}$ is given by
\begin{align*}
     & \widehat{\mc N}_{s,t}\begin{pmatrix}
         \rho_{00} & \rho_{01} \\
         \rho_{10} & \rho_{11}
     \end{pmatrix} = \begin{pmatrix}
        s \rho_{00} & 0 & \sqrt{s} \rho_{01} \\
        0 & (1-s-t)\rho_{00} & 0 \\
        \sqrt{s} \rho_{10} & 0 & t\rho_{00} + \rho_{11}
    \end{pmatrix},   \\
    & \widehat{\mc N}^c_{s,t}\begin{pmatrix}
         \rho_{00} & \rho_{01} \\
         \rho_{10} & \rho_{11}
     \end{pmatrix} = \begin{pmatrix}
        s \rho_{00} + \rho_{11} & 0& \sqrt{t} \rho_{10} \\
        0 & (1-s-t)\rho_{00} & 0 \\
        \sqrt{t} \rho_{01} & 0 & t\rho_{00} 
    \end{pmatrix}.
\end{align*}
In terms of Kraus representation, we have $\widehat{\mc N}_{s,t}(\rho) = \sum_{k=0}^3 E^{Res}_k \rho (E^{Res}_k)^{\dagger},\ \mc N_{s,t}^c(\rho) = \sum_{k=0}^3 \wt E^{Res}_k \rho  (\wt E^{Res}_k)^{\dagger}$, where 
\begin{align*}
    &E^{Res}_0 = \sqrt{s} \ketbra{0}{0} + \ketbra{2}{1},\ E^{Res}_1 = \sqrt{1-s-t} \ketbra{1}{0},\ E^{Res}_2 = \sqrt{t} \ketbra{2}{0}
\end{align*}
and
\begin{align*}
    &\wt E^{Res}_0 = \sqrt{s} \ketbra{0}{0},\ \wt E^{Res}_1 = \sqrt{1-s-t} \ketbra{1}{0},\ \wt E^{Res}_2 = \sqrt{t} \ketbra{2}{0} + \ketbra{0}{1}.
\end{align*}
In terms of transfer matrix, we have 
    \begin{align*}
        & \mc T_{\widehat{\mc N}_{s,t}} = \begin{pmatrix}
            s & 0 & 0 & 0  \\
            0 & 0 & 0 & 0 \\
            0 & \sqrt{s} & 0 & 0  \\
            0 & 0 & 0 & 0  \\
            1-s-t & 0 & 0 & 0   \\
            0 & 0 & 0 & 0  \\
            0 & 0 & \sqrt{s} & 0  \\
            0 & 0 & 0 & 0  \\
            t & 0 & 0 & 1   \\
        \end{pmatrix}, \
    \mc T_{\widehat{\mc N}^c_{s,t}} = \begin{pmatrix}
            s & 0 & 0 & 1  \\
            0 & 0 & 0 & 0 \\
            0 & 0 & \sqrt{t} & 0  \\
            0 & 0 & 0 & 0  \\
            1-s-t & 0 & 0 & 0   \\
            0 & 0 & 0 & 0  \\
            0 & \sqrt{t} & 0 & 0  \\
            0 & 0 & 0 & 0  \\
            t & 0 & 0 & 0   \\
        \end{pmatrix}.
\end{align*}
It is straightforward to see $\mc T_{\mc D} \mc T_{\widehat{\mc N}_{s,t}} = \mc T_{\widehat{\mc N}^c_{s,t}}$, where
\begin{align*}
\mc T_{\mc D} = \begin{pmatrix}
            1- \frac{t}{s} & 0 & 0 & 0 & 0 & 0 & 0 & 0 & 1  \\
            0 & 0 & 0 & 0 & 0 & 0 & 0 & 0 & 0  \\
            0 & 0 & 0 & 0 & 0 & 0 & \frac{\sqrt{t}}{\sqrt{s}} & 0 & 0  \\
            0 & 0 & 0 & 0 & 0 & 0 & 0 & 0 & 0  \\
            0 & 0 & 0 & 0 & 1 & 0 & 0 & 0 & 0  \\
            0 & 0 & 0 & 0 & 0 & 0 & 0 & 0 & 0  \\
            0 & 0 & \frac{\sqrt{t}}{\sqrt{s}} & 0 & 0 & 0 & 0 & 0 & 0  \\
            0 & 0 & 0 & 0 & 0 & 0 & 0 & 0 & 0  \\
            \frac{t}{s} & 0 & 0 & 0 & 0 & 0 & 0 & 0 & 0  \\
        \end{pmatrix}.
\end{align*}
Note that if $t \le s$, $\mc T_{\mc D}$ induces a quantum channel, which implies $\widehat{\mc N}_{s,t}$ is degradable. Similarly, we have $\mc T_{\wt {\mc D}} \mc T_{\widehat{\mc N}^c_{s,t}} = \mc T_{\widehat{\mc N}_{s,t}}$, where
\begin{align*}
\mc T_{\wt {\mc D}} = \begin{pmatrix}
            0 & 0 & 0 & 0 & 0 & 0 & 0 & 0 & \frac{s}{t}  \\
            0 & 0 & 0 & 0 & 0 & 0 & 0 & 0 & 0  \\
            0 & 0 & 0 & 0 & 0 & 0 & \frac{\sqrt{s}}{\sqrt{t}} & 0 & 0  \\
            0 & 0 & 0 & 0 & 0 & 0 & 0 & 0 & 0  \\
            0 & 0 & 0 & 0 & 1 & 0 & 0 & 0 & 0  \\
            0 & 0 & 0 & 0 & 0 & 0 & 0 & 0 & 0  \\
            0 & 0 & \frac{\sqrt{s}}{\sqrt{t}} & 0 & 0 & 0 & 0 & 0 & 0  \\
            0 & 0 & 0 & 0 & 0 & 0 & 0 & 0 & 0  \\
            1 & 0 & 0 & 0 & 0 & 0 & 0 & 0 & 1 - \frac{s}{t} \\
        \end{pmatrix}.
\end{align*}
Note that if $s \le t$, $\mc T_{\wt {\mc D}}$ induces a quantum channel, which implies $\widehat{\mc N}_{s,t}$ is anti-degradable. 

\textbf{Case II: $s \le 1-s-t$.} 
Via exactly the same argument as \textbf{Case I}, we see that if $t \le 1-s-t$, $ \widehat{\mc N}_{s,t}$ is degradable and, if $t \ge 1-s-t$, $ \widehat{\mc N}_{s,t}$ is anti-degradable. 
\end{proof}
The degradable and anti-degradable regions are summarized in Figure \ref{figure:deg region platypus}:

\begin{figure}[H]
\centering
\begin{tikzpicture}[scale=6]
    \draw[->] (0,0) -- (1,0) node[right] {$t$}; 
    \draw[->] (0,0) -- (0,1) node[above] {$s$}; 
    
    \draw[dotted] (0.5,0) -- (0.5,1) node[above right] {$t = 0.5$};
    
    \fill[blue!20, opacity=0.5] (0.33,0.33) -- (0.5,0.5) -- (0.5,0) -- cycle; 
    \fill[red!20, opacity=0.5] (0,0) -- (0,1) -- (0.5,0) -- cycle;
    \fill[red!20, opacity=0.5] (0.33,0.33) -- (0,1) -- (0.5,0.5) -- cycle;
    \draw[scale=1, domain=0:0.5, smooth, variable=\x, black, dotted] plot ({\x}, {\x}); 
    \draw[scale=1, domain=0:0.5, smooth, variable=\x, black, dotted] plot ({\x}, {1 - 2*\x}); 
    \draw[scale=1, domain=0:0.5, smooth, variable=\x, black, dotted] plot ({\x}, {1 - \x}); 

    \node[blue] at (0.7,0.25) {Blue region: antidegradable};
    \node[red] at (0.3,0.6) {Red region: degradable};
\end{tikzpicture}
\caption{Degradability and antidegradability regions for $\widehat{\mc N}_{s,t}$.}
\label{figure:deg region platypus}
\end{figure}

It helps us plot the maximal coherent information of $\mc N_{s,t}$ in the region where it is neither degradable nor anti-degradable:
\begin{figure}[ht]
    \centering\includegraphics[width=.5\textwidth]{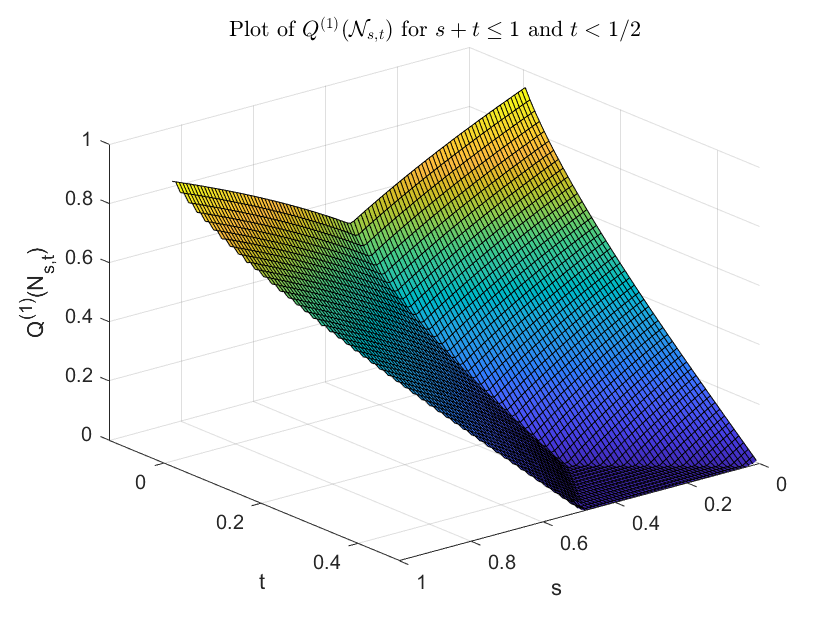}
    \caption{Plot of the coherent information of $\mc N_{s,t}$}
    \label{fig:Plot of Coherent information}
\end{figure}

\begin{remark}
    Note that our argument also establishes \emph{conjugate non-degradability} \cite{bradler2010conjugate}, since our channel maps a real matrix to a real matrix. By restricting to real matrices, the conjugation of the complementary channel is identical to the original complementary channel. Consequently, there exists no degrading map to the conjugation of the complementary channel.
\end{remark}
\begin{remark}
    It is easy to see that the channel $\mc N_{s,t}$ and $\widehat{\mc N}_{s,t}$ are not PPT, unless $t = 1$. 
\end{remark}

\subsection{Additivity property} 
\subsubsection{Strong and weak additivity when $s+t = 1$ or $s=0$}
\begin{prop}\label{main:strong additivity}
    Suppose $s+t = 1$ or $s=0$. We have $\mc Q(\mc N_{s,t}) = \mc Q^{(1)}(\mc N_{s,t})$ and for any degradable channel $\mc M$, we have 
\begin{equation}
    \mc Q^{(1)}(\mc N_{s,t} \otimes \mc M ) = \mc Q^{(1)}(\mc N_{s,t}) + \mc Q^{(1)}(\mc M).
\end{equation}
\end{prop}
\begin{proof}
When $s+t = 1$, there exists a qutrit-to-qubit quantum channel $\mc A$ defined by 
\begin{equation}
    \mc A(\sum_{i,j=0}^2 \rho_{ij} \ketbra{i}{j}) = \begin{pmatrix}
    \rho_{00} & \rho_{01} \\
    \rho_{10} & \rho_{11}+ \rho_{22}
\end{pmatrix},
\end{equation}
such that $\widehat{\mc N}_{s,t} \circ \mc A = \mc N_{s,t}$. Then using $\mc Q^{(1)}(\widehat{\mc N}_{s,t}) = \mc Q^{(1)}(\mc N_{s,t})$ and Theorem \ref{thm:additivity via DPI}, we get the desired additivity properties. For the case $s = 0$, the construction is similar and we conclude the proof. 
\end{proof}

\subsubsection{Discussion of weak additivity when $s+t < 1$ and $s > 0$}  
In this region, weak additivity may still hold. Specifically, we showed that \(\mc N_{s,t}\) is weakly dominated by its restriction channel, in the sense of Definition \ref{def:weak_dom}. The restriction channel is either degradable or anti-degradable. Although in this case \(\mc N_{s,t}\) cannot be simulated by its restriction channel (meaning Theorem \ref{thm:additivity via DPI} is not applicable), it appears that entangled inputs do not help boost the coherent information when many copies of \(\mc N_{s,t}\) are used. 

If weak additivity were rigorously verified for any \(s > 0\) and \(s + t < 1\), then \(\mc N_{s,t}\) would have zero quantum capacity in certain regions, despite being neither anti-degradable nor PPT.

When \(t = 0\), weak additivity can be verified if the \emph{spin alignment conjecture} holds \cite{leditzky2023platypus}. An additional observation from \cite{leditzky2023platypus} is that the structure of \(\mc N_{s,0}^{\otimes n}\) for any \(n\) suggests the optimizer of coherent information must take a simple form. This claim is closely tied to their spin alignment conjecture, which implies precisely that simple form of the optimizer. However, since the optimization is non-convex and spans arbitrary \(n\), it is not straightforward to derive this form directly. 

Currently, the best known result (see \cite{alhejji2024towards,alhejji2024refining}) is
\begin{equation}
        \mc Q^{(1)}(\mc N_{s,0}^{\otimes 2}) \;=\; 2\, \mc Q^{(1)}(\mc N_{s,0}).
\end{equation}
We leave the rigorous proof of weak additivity as an open problem.

\subsection{Failure of strong additivity with degradable channels}\label{sec:amplification}
In this subsection, we demonstrate that for the parameter regime \( s + t < 1 \) and \( s > 0 \), strong additivity fails when \(\mc N_{s,t}\) tensors with degradable channels . Recall that \(\mc Q^{(1)}(\mc N_{s,t}) = \mc Q^{(1)}(\widehat{\mc N}_{s,t})\), where the channel \(\widehat{\mc N}_{s,t}\), defined by \eqref{def:restriction channel}, is either degradable or anti-degradable depending on the parameter values.

To establish the failure of strong additivity, we employ two distinct arguments:
\begin{enumerate}
    \item Smith-Yard argument \cite{smith2008quantum} in the region where \(\widehat{\mc N}_{s,t}\) is anti-degradable.
    \item Log-singularity argument \cite{siddhu2021entropic} in the region where \(\widehat{\mc N}_{s,t}\) is degradable.
\end{enumerate}

It is important to note that neither of these methods alone suffices for both cases, necessitating the combined approach.

\subsubsection{Anti-degradable region---Smith-Yard argument}
The Smith-Yard argument, introduced in \cite{smith2008quantum}, shows that strong non-additivity can be achieved for erasure channels with an erasure probability of \(1/2\), provided that the maximal private information of a channel exceeds twice its maximal coherent information.

We define a \(d\)-dimensional erasure channel with erasure probability \(\lambda \in [0,1]\) as follows:
\begin{equation}\label{def:erasure channel}
    \mc E_{d,\lambda}(\rho) := (1-\lambda)\rho + \lambda \ketbra{e}{e},
\end{equation}
where \(\ket{e}\) is a fixed erasure state orthogonal to the input state space.

For any ensemble of states \(\{p_x, \rho_x^A\}_{x \in \mc X}\) and a channel \(\mathcal{N}\) with input \(A\), output \(B\), and environment \(E\), there exists a system \(C\) (of dimension equal to the sum of the ranks of the states \(\rho_x^A\)) and a joint state \(\rho^{AC}\) such that
\begin{align*}
   & I_c(\rho^{AC}, \mathcal{N} \otimes \mathcal{E}_{d_C, 1/2}) = \frac{1}{2} \bigl(I(\mc X; B) - I(\mc X; E)\bigr) 
      \;=\; \frac{1}{2} I_p(\{p_x, \rho_x^A\}, \mathcal{N}),
\end{align*}
where the private information of the channel \(\mathcal{N}\) for the ensemble \(\{p_x, \rho_x^A\}\) is defined as:
\begin{align*}
    I_p(\{p_x, \rho_x^A\}, \mathcal{N}) := I(\mc X; B) \;-\; I(\mc X; E),\quad 
    \mc P^{(1)}(\mc N) := \sup_{\{p_x, \rho_x^A\}} I_p(\{p_x, \rho_x^A\}, \mathcal{N}).
\end{align*}

If the maximal private information of \(\mc N\) exceeds twice its maximal coherent information, i.e.,
\[
\mc P^{(1)}(\mc N) > 2 \,\mc Q^{(1)}(\mc N),
\]
then it can be shown that
\begin{equation}
    \mc Q^{(1)}\bigl(\mc N \otimes \mathcal{E}_{d_C, 1/2}\bigr) > \mc Q^{(1)}(\mc N),
\end{equation}
for sufficiently large \(d_C\).

Applying this technique to \(\widehat{\mc N}_{s,t}\), we establish the failure of strong additivity in the anti-degradable region. Specifically, this demonstrates that \(\mc Q^{(1)}(\mc N_{s,t})\) is non-additive when combined with suitable erasure channels.
\begin{prop}\label{main:strong non-additivity Smith-Yard}
    Suppose $s>0$, $s+t < 1$ and $\mc Q^{(1)}(\mc N_{s,t})=0$. Then 
\begin{equation}
    \mc Q^{(1)}(\mc N_{s,t} \otimes \mc E_{d,\frac{1}{2}}) > \mc Q^{(1)}(\mc N_{s,t}), \quad d \ge 3.
\end{equation}
\end{prop}
\begin{proof}
We can choose 
\begin{equation}
    \rho_1^A = \ketbra{0}, \quad \rho_2^A = u\ketbra{1} + (1-u)\ketbra{2} 
\end{equation}
and optimize over $\{p_x, \rho_x^A\}_{x = 1,2}$, we see $\mc P^{(1)}(\mc N_{s,t})>0$ for any $t<\frac{1}{2}$, see Figure \ref{fig:private info}. In this case, $d \ge 3$ suffices to observe non-additivity. 
\begin{figure}[ht]
    \centering\includegraphics[width=.5\textwidth]{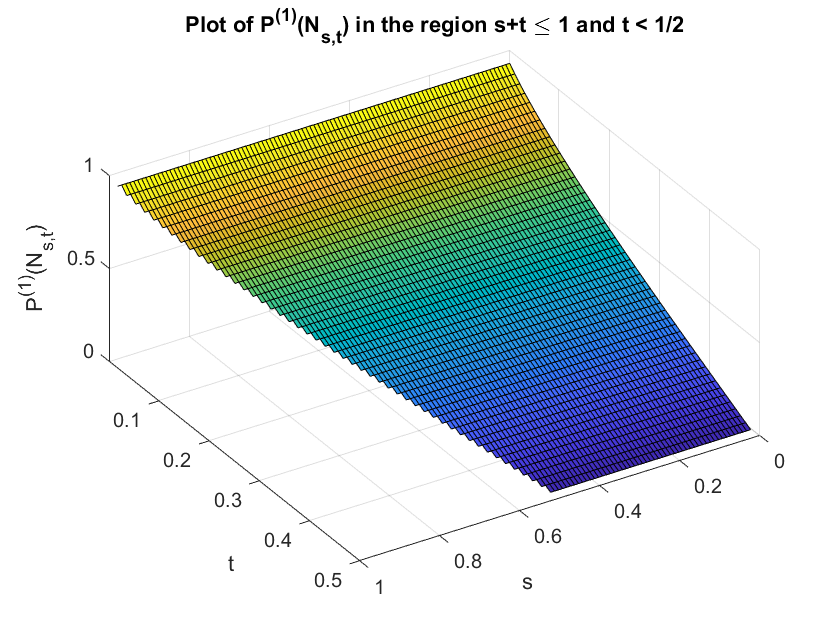}
    \caption{Plot of the private information of $\mc N_{s,t}$}
    \label{fig:private info}
\end{figure}
\end{proof}

\subsubsection{Degradable region---Log-singularity argument}
Using the idea of \(\varepsilon\)\emph{-log-singularity}~\cite{siddhu2021entropic}, we present a framework exhibiting amplification of coherent information. In other words, we discuss how
\[
   \mc Q^{(1)}\bigl(\mc N_1 \otimes \mc N_2\bigr) 
   \;>\; \mc Q^{(1)}(\mc N_1) \;+\; \mc Q^{(1)}(\mc N_2)
\]
can occur. Suppose \(\rho_{A_1}, \rho_{A_2}\) are the optimizers of \(\mc Q^{(1)}(\mc N_1)\) and \(\mc Q^{(1)}(\mc N_2)\), respectively; that is,
\[
    I_c(\rho_{A_1}, \mc N_1) \;=\; \mc Q^{(1)}(\mc N_1),
    \quad 
    I_c(\rho_{A_2}, \mc N_2) \;=\; \mc Q^{(1)}(\mc N_2).
\]
Then the tensor product state \(\rho_{A_1} \otimes \rho_{A_2}\) already achieves a lower bound for \(\mc Q^{(1)}(\mc N_1 \otimes \mc N_2)\) that is at least
\(\mc Q^{(1)}(\mc N_1) + \mc Q^{(1)}(\mc N_2)\).

Moreover, if we perturb \(\rho_{A_1}\otimes \rho_{A_2}\) using an entangled state \(\sigma_{A_1A_2}\), \emph{i.e.},
\begin{equation}
    \rho_{A_1A_2}(\varepsilon)
    \;:=\; 
    (1-\varepsilon)\,\rho_{A_1}\otimes \rho_{A_2} 
    \;+\; 
    \varepsilon \,\sigma_{A_1A_2},
\end{equation}
we can potentially achieve a larger coherent information as \(\varepsilon \to 0\). The underlying reason is that Lipschitz continuity for the von Neumann entropy fails in the presence of a possible log-singularity. This phenomenon was first noted by Fannes~\cite{fannes1973continuity} and further sharpened in~\cite{petz2007quantum}. More recently, the logarithmic dimension factor has been refined in~\cite{berta2024continuity,audenaert2024continuity}.

\begin{lemma}
\label{lemma:continuity}
For density operators \(\rho\) and \(\sigma\) on a Hilbert space \(\mc H_A\) of finite dimension \(d_A\), if \(\tfrac12\|\rho-\sigma\|_1 \le \varepsilon \le 1\), then
\[
\bigl|S(\rho) - S(\sigma)\bigr|
\;\le\;
\begin{cases}
   \varepsilon \log (d_A - 1) + h(\varepsilon) 
   & \text{if } \varepsilon \leq 1 - \frac{1}{d_A},\\[6pt]
   \log d_A 
   & \text{if } \varepsilon > 1 - \frac{1}{d_A},
\end{cases}
\]
where \(h(x) = -x \log x - (1 - x)\log (1-x)\) is the binary entropy.
\end{lemma}

This continuity result motivates the following definition.

\begin{definition}[\(\varepsilon\)-log-singularity]
\label{def:log-singularity}
Let \(\{\sigma(\varepsilon)\}_{\varepsilon \ge 0}\) be a family of density operators that depends on \(\varepsilon\), and suppose there exists a universal constant \(C > 0\) such that
\(\|\sigma(\varepsilon) - \sigma(0)\|_1 \le C \varepsilon\).
We say that \(\sigma(\varepsilon)\) has an \(\varepsilon\)-log-singularity \emph{of rate} \(r \in (-\infty, +\infty)\) if
\[
 \lim_{\varepsilon \to 0^+} \frac{S(\sigma(\varepsilon)) - S(\sigma(0))}{\varepsilon \,\bigl|\log \varepsilon\bigr|} 
 \;=\; r.
\]
\end{definition}

\noindent The following example summarizes different types of perturbations and their corresponding rates of \(\varepsilon\)-log-singularity.

\begin{example}\label{example:log-singularity}
Suppose \(\ket{\varphi}, \ket{\psi}\) are two orthogonal pure states, and let \(\rho_0\) be a density operator whose support is orthogonal to \(\ket{\varphi}\) and \(\ket{\psi}\). Assume \(a \in (0,1)\) and \(b > 0\). Then the \(\varepsilon\)-log-singularity rates are computed as follows:
\begin{enumerate}
    \item \label{log-singularity:1}
    If \begin{align*}
        \sigma(\varepsilon) =  a \ketbra{\varphi}{\varphi} + (1-a)\rho_0  - b\varepsilon \ketbra{\varphi}{\varphi} +  b\varepsilon \ketbra{\psi}{\psi},
    \end{align*}
    then \(\sigma(\varepsilon)\) has an \(\varepsilon\)-log-singularity of rate \(b > 0\).
    \item \label{log-singularity:2}
    If \begin{align*}
        \sigma(\varepsilon) & = a \ketbra{\varphi}{\varphi} + (1-a)\rho_0 - b\varepsilon \ketbra{\varphi}{\varphi} + b \varepsilon \ketbra{\psi}{\psi} + b \sqrt{\varepsilon(1-\varepsilon)} (\ketbra{\psi}{\varphi} + \ketbra{\varphi}{\psi}) ,
    \end{align*}
    then \(\sigma(\varepsilon)\) has an \(\varepsilon\)-log-singularity of rate \(\tfrac{b(a-b)}{a}\).

    \item \label{log-singularity:3}
    (No log-singularity) If 
    \[
       \sigma(\varepsilon)
       \;=\;
       \sigma(0) 
       \;+\; \varepsilon\,H,
    \]
    where \(H\) is a Hermitian, traceless operator such that $\text{supp}(H) \subseteq \text{supp}(\sigma(0))$, then \(\sigma(\varepsilon)\) has an \(\varepsilon\)-log-singularity of rate \(0\).
\end{enumerate}
\end{example}
Now we discuss the channel setting. Suppose $U_{\mc N_1}: \mc H_{A_1} \to \mc H_{B_1} \otimes \mc H_{E_1}, \quad U_{\mc N_2}: \mc H_{A_2} \to \mc H_{B_2} \otimes \mc H_{E_2}$ are two isometries and $(\mc N_1,\mc N_1^c), (\mc N_2,\mc N_2^c)$ are the two complementary pairs of quantum channels generated by $U_{\mc N_1}, U_{\mc N_2}$ respectively. Denote $\sigma(0) = \rho_{A_1}\otimes \rho_{A_2}$, where $\rho_{A_1}, \rho_{A_2}$ are the optimizers of $\mc Q^{(1)}(\mc N_1), \mc Q^{(1)}(\mc N_2)$. We choose a perturbation $\sigma(\varepsilon) = \sigma(0) + \varepsilon H$ of $\sigma(0)$, where $H$ is a traceless Hermitian operator and $\sigma(\varepsilon)$ is an entangled state. Denote 
\begin{equation}
\begin{aligned}
    & \rho_{B_1B_2}(\varepsilon) = (\mc N_1 \otimes \mc N_2)(\sigma(\varepsilon)), \\
    & \rho_{E_1E_2}(\varepsilon) = (\mc N_1^c \otimes \mc N_2^c)(\sigma(\varepsilon)).
\end{aligned}
\end{equation}
Then for any $\varepsilon > 0$ such that $\sigma(\varepsilon)$ is a state,
\begin{equation}\label{framework:amplification}
    \begin{aligned}
   & \mc Q^{(1)}(\mc N_1 \otimes \mc N_2) - (\mc Q^{(1)}(\mc N_1) + \mc Q^{(1)}(\mc N_2)) \ge I_c(\sigma(\varepsilon), \mc N_1 \otimes \mc N_2) - I_c(\sigma(0), \mc N_1 \otimes \mc N_2)\\
   & = S(\rho_{B_1B_2}(\varepsilon)) - S(\rho_{B_1B_2}(0)) - \big(S(\rho_{E_1E_2}(\varepsilon)) - S(\rho_{E_1E_2}(0))\big).
\end{aligned}
\end{equation}
Denote $\Delta_B(\varepsilon) = S(\rho_{B_1B_2}(\varepsilon)) - S(\rho_{B_1B_2}(0))$, $\Delta_E(\varepsilon) = S(\rho_{E_1E_2}(\varepsilon)) - S(\rho_{E_1E_2}(0))$. If 
\begin{align}\label{amplification:sufficient}
    \lim_{\varepsilon\to 0+} \frac{\Delta_B(\varepsilon) - \Delta_E(\varepsilon)}{\varepsilon |\log \varepsilon|} =: r_B - r_E >0.
\end{align}
Then choose $\varepsilon>0$ reasonably small we have $\mc Q^{(1)}(\mc N_1 \otimes \mc N_2) - (\mc Q^{(1)}(\mc N_1) + \mc Q^{(1)}(\mc N_2)) >0$. 

We can formally show the following using \eqref{amplification:sufficient}:
\begin{prop}\label{main:strong non-additivity}
    Suppose $s>0$, $s+t < 1$ and $\mc Q^{(1)}(\mc N_{s,t})>0$. Then 
\begin{equation}
    \mc Q^{(1)}(\mc N_{s,t} \otimes \mc E_{2,\lambda}) > \mc Q^{(1)}(\mc N_{s,t}), 
\end{equation}
if $\lambda$ satisfies
\begin{equation}\label{inequality:lambda region}
    \frac{1}{2} \le \lambda < \begin{cases}
         \frac{1 - (s+t)u_*}{1 + u_* - 2(s+t)u_*},\ s \le \frac{1-t}{2}, \\
         \frac{1 - (1-s)u_*}{1 + u_* - 2(1-s)u_*},\ s \ge \frac{1-t}{2}.
    \end{cases}
\end{equation}
\added{where $u_* = u_*(s,t) \in (0,1)$ is the optimal parameter achieving the maximal coherent information of $ \mc N_{s,t}$ in \eqref{eqn:optimizer of generalized platypus}:
\begin{align*}
  \mc Q^{(1)}(\mc N_{s,t}) = 
  \begin{cases}
        I_c(u_* \ketbra{0}{0} + (1-u_*) \ketbra{2}{2}, \mc N_{s,t} ), \quad 1-s-t \ge s, \\
       I_c(u_* \ketbra{0}{0} + (1-u_*) \ketbra{1}{1}, \mc N_{s,t} ), \quad 1-s-t \le s. \\
  \end{cases}
\end{align*}}
\end{prop}
\begin{proof}
Denote the isometry of erasure channel as $\mc U_{\mc E_{2,\lambda}}: \mc H_{A'} \to \mc H_{B'} \otimes \mc H_{E'}$ with $\text{dim} \mc H_{B'} = \text{dim} \mc H_{E'} = 3$:
\begin{equation}\label{eqn:erasure isomtery}
    \begin{aligned}
     & \mc U_{\mc E_{2,\lambda}}\ket{0} = \sqrt{1-\lambda} \ket{02} + \sqrt{\lambda} \ket{20}, \\
     & \mc U_{\mc E_{2,\lambda}}\ket{1} = \sqrt{1-\lambda} \ket{12} + \sqrt{\lambda} \ket{21},  
\end{aligned}
\end{equation}
where $\ket{2} = \ket{e}$ is the erasure flag. 

\textbf{Case I: $s \le 1-s-t$}. In this case, under the assumption $\lambda \ge \frac{1}{2}$, we can choose $\ketbra{0}$ as the optimal input state of $\mc Q^{(1)}(\mc E_{2,\lambda})$. Then the product state 
\begin{equation}
   \rho(0) = u_* \ketbra{00} + (1-u_*) \ketbra{20} \in \mb B(\mc H_A \otimes \mc H_{A'})
\end{equation}
achieves $\mc Q^{(1)}(\mc N_{s,t}) + \mc Q^{(1)}(\mc E_{2,\lambda})$, i.e., $I_c(\rho(\varepsilon), \mc N_{s,t} \otimes \mc E_{2,\lambda}) = \mc Q^{(1)}(\mc N_{s,t}) + \mc Q^{(1)}(\mc E_{2,\lambda})$. Note that $\ket{1}_A, \ket{1}_{A'}$ is not used in the optimization of $\mc Q^{(1)}(\mc N_{s,t})$ and $\mc Q^{(1)}(\mc E_{2,\lambda})$, we aim to achieve amplification using $\ket{11}$. To this end, denote the entangled input state $\rho(\varepsilon) \in \mb B(\mc H_A \otimes \mc H_{A'})$ as 
\begin{equation}
    \rho(\varepsilon) = u_* \ketbra{00} + (1-u_*) \ketbra{\psi_{\varepsilon}},
\end{equation}
where $\ket{\psi_{\varepsilon}} = \sqrt{1-\varepsilon} \ket{20} + \sqrt{\varepsilon} \ket{11}$ and denote 
\begin{equation}
    \begin{aligned}
     & \rho_{BB'}(\varepsilon) = (\mc N_{s,t} \otimes \mc E_{2,\lambda})(\rho(\varepsilon)), \\
     &\rho_{EE'}(\varepsilon) = (\mc N^c_{s,t} \otimes \mc E^c_{2,\lambda})(\rho(\varepsilon)).
    \end{aligned}
\end{equation}
Following the framework of $\varepsilon$-log-singularity \eqref{framework:amplification}, if we show that $\rho_{BB'}(\varepsilon)$ has a higher rate of $\varepsilon$-log-singularity than $\rho_{EE'}(\varepsilon)$, then we have $\mc Q^{(1)}(\mc N_{s,t} \otimes \mc E_{2,\lambda}) > \mc Q^{(1)}(\mc N_{s,t}) + \mc Q^{(1)}(\mc E_{2,\lambda}) = \mc Q^{(1)}(\mc N_{s,t})$. In fact, using the expression of $\mc N_{s,t}$ \eqref{eqn:Generalized Platypus} and $\mc E_{2,\lambda}$ \eqref{eqn:erasure isomtery}, we have 
\begin{equation}\label{eqn:case I log singularity}
    \begin{aligned}
        & \rho_{BB'}(0)  = \left(u_*s \ketbra{0} 
        + u_*(1-s-t)\ketbra{1} 
        + (u_* t + 1-u_*) \ketbra{2} \right) \otimes ((1-\lambda)\ketbra{0} + \lambda \ketbra{2}), \\
        & \rho_{BB'}(\varepsilon)  = \rho_{BB'}(0) 
        + (1-u_*)(1-\lambda) \varepsilon 
        \left(\ketbra{21} - \ketbra{20} \right), \\
        & \rho_{EE'}(0) = \big(u_*s \ketbra{0} 
        + \big(u_*(1-s-t) + (1-u_*)\big)\ketbra{1} 
        + u_* t\ketbra{2} \big) \otimes (\lambda \ketbra{0} + (1-\lambda)\ketbra{2}), \\
        & \rho_{EE'}(\varepsilon)  = \rho_{EE'}(0) 
        + (1-u_*)(1-\lambda) \varepsilon 
        \left(\ketbra{02} - \ketbra{12}\right) \\
        & \hspace{1.5cm} + (1-u_*)\lambda 
        \big[\varepsilon \ketbra{01} 
        - \varepsilon \ketbra{10}  + \sqrt{\varepsilon(1-\varepsilon)} 
        (\ketbra{10}{01} + \ketbra{01}{10})\big].
    \end{aligned}
\end{equation}
Using Example \ref{example:log-singularity} \eqref{log-singularity:1}, $\rho_{BB'}(\varepsilon)$ has an $\varepsilon$-log-singularity of rate $(1-u_*)(1-\lambda)>0$. Note that for the state $\rho_{EE'}(\varepsilon)$, since $s>0, s+t<1$, $\rho_{EE'}(0)$ has full support on $\ket{02}, \ket{12}$ thus the $\varepsilon$-perturbation on that subspace does not have $\varepsilon$-log-singularity. Therefore, using Example \ref{example:log-singularity} \eqref{log-singularity:2}, \eqref{log-singularity:3}, $\rho_{EE'}(\varepsilon)$ has an $\varepsilon$-log-singularity of rate
\begin{align*}
    & \frac{b(a-b)}{a} = \frac{\lambda u_*(1-u_*)(1-s-t)}{1- (s+t)u_*}, \quad  a = \lambda (u_*(1-s-t)+(1-u_*)),\ b = (1-u_*)\lambda.
\end{align*}
Via \eqref{amplification:sufficient}, we have $\mc Q^{(1)}(\mc N_{s,t} \otimes \mc E_{2,\lambda}) > \mc Q^{(1)}(\mc N_{s,t})$ if $\lambda \ge \frac{1}{2}$ and 
\begin{align*}
       & (1-u_*)(1-\lambda) > \frac{\lambda u_*(1-u_*)(1-s-t)}{1- (s+t)u_*} \iff \lambda < \frac{1 - (s+t)u_*}{1 + u_* - 2(s+t)u_*} \in (\frac{1}{2}, 1). 
\end{align*}

\textbf{Case II: $s \ge 1-s-t$}. In this case, similar as before, using \eqref{eqn:optimizer of generalized platypus}, we can choose the product state 
\begin{equation}
   \rho(0) = u_* \ketbra{00} + (1-u_*) \ketbra{10} \in \mb B(\mc H_A \otimes \mc H_{A'})
\end{equation}
Note that $\ket{2}_A, \ket{1}_{A'}$ is not used in the optimization of $\mc Q^{(1)}(\mc N_{s,t})$ and $\mc Q^{(1)}(\mc E_{2,\lambda})$, we aim to achieve amplification using $\ket{21}$. 
\begin{equation}
    \rho(\varepsilon) = u_* \ketbra{00} + (1-u_*) \ketbra{\psi_{\varepsilon}},
\end{equation}
where $\ket{\psi_{\varepsilon}} = \sqrt{1-\varepsilon} \ket{10} + \sqrt{\varepsilon} \ket{21}$. Using the same notation, similar calculation shows that 
\begin{equation}\label{eqn:case II log singularity}
    \begin{aligned}
        & \rho_{BB'}(0) = \left(u_*s \ketbra{0} 
        + u_*(1-s-t)\ketbra{1} 
        + (u_* t + 1-u_*) \ketbra{2} \right) \otimes ((1-\lambda)\ketbra{0} + \lambda \ketbra{2}), \\
        & \rho_{BB'}(\varepsilon) = \rho_{BB'}(0) 
        + (1-u_*)(1-\lambda) \varepsilon 
        \left(\ketbra{21} - \ketbra{20} \right), \\
        & \rho_{EE'}(0) = \big((u_*s + 1-u_*)\ketbra{0} 
        + u_*(1-s-t)\ketbra{1} 
        + u_* t\ketbra{2} \big) \otimes (\lambda \ketbra{0} + (1-\lambda)\ketbra{2}), \\
        & \rho_{EE'}(\varepsilon) = \rho_{EE'}(0) 
        + (1-u_*)(1-\lambda) \varepsilon 
        \left(\ketbra{12} - \ketbra{02} \right) \\
        & \hspace{1.5cm} + (1-u_*)\lambda 
        \big[\varepsilon \ketbra{11} 
        - \varepsilon \ketbra{00} + \sqrt{\varepsilon(1-\varepsilon)} 
        (\ketbra{00}{11} + \ketbra{11}{00})\big].
    \end{aligned}
\end{equation}
Using Example \ref{example:log-singularity} and the same argument, $\rho_{BB'}(\varepsilon)$ has an $\varepsilon$-log-singularity of rate $(1-u_*)(1-\lambda)>0$ and $\rho_{EE'}(\varepsilon)$ has an $\varepsilon$-log-singularity of rate $\frac{\lambda u_*(1-u_*)s}{1- (1-s)u_*}$. Similar as before, via \eqref{amplification:sufficient}, we have $\mc Q^{(1)}(\mc N_{s,t} \otimes \mc E_{2,\lambda}) > \mc Q^{(1)}(\mc N_{s,t})$ if 
\begin{align*}
    \frac{1}{2} \le \lambda < \frac{1 - (1-s)u_*}{1 + u_* - 2(1-s)u_*}.
\end{align*}
\end{proof}
\begin{remark}
    Note that non-additivity can still happen when $\lambda$ is outside the region \eqref{inequality:lambda region}. In fact, when $\lambda = \frac{1}{2}$, we achieve non-additivity. Then, by continuity of coherent information \cite{Leung_2009}, $\lambda$ can be extended to values smaller than $\frac{1}{2}$ while preserving non-additivity.  
\end{remark}

\section{Conclusion and open problems}
In this work, we have systematically studied the conditions under which a quantum channel can exhibit weak or strong additive coherent information. Our investigation reveals two classes of channels that retain additivity. The first class of channels encompasses fundamental building blocks such as: \begin{itemize} \item \emph{(Weak) degradable channels} (Definition \ref{def:weaker degradability}). \item Direct sums and tensor products of degradable and anti-degradable channels. \item PPT channels $\Gamma$ and the identity map tensored with PPT channels. \end{itemize} It is known that these channels often exhibit weak additivity of coherent information, and for PPT class, strong additivity can fail when degradable channels are involved \cite{smith2008quantum}. Beyond these basic examples, via Theorem~\ref{thm:additivity via DPI}, we can construct new classes of channels whose structure may obscure the common structure like degradability and PPT, but it still retain either weak or strong additivity. 

Our findings also highlight that the interplay between \emph{strong} and \emph{weak} additivity is rich and subtle. Strong additivity requires that a channel maintain additivity when composed with any other channel from a subclass, whereas weak additivity demands only additivity under repeated uses of the same channel. By revealing classes of channels that meet one or both of these criteria—some of which are non-degradable and non-PPT—this work expands the known landscape where additivity holds. 

Finding a systematic way to show weak additivity with positive capacity, while strong additivity with degradable channels fails remains an interesting open question.
\appendices
\counterwithin*{equation}{section}
\renewcommand\theequation{\thesection\arabic{equation}}

\section{Equality case for amplitude damping channels} \label{appendix:equality case}
\noindent \textbf{Proposition \ref{main:uniqueness of fixed points}. }Suppose $\mc N = \mc N^{B\to E}$ has a unique fixed state, i.e., there exists a unique quantum state $\rho_0$ such that $\mc N(\rho_0) = \rho_0$. Then for any finite-dimensional quantum system $\mc H_V$ and quantum state $\rho_{VB}$, 
    \begin{equation*}
        \left(id_{\mb B(\mc H_V)} \otimes \mc N\right) (\rho_{VB}) = \rho_{VB}
    \end{equation*}
    if and only if $\rho_{VB} = \rho_V \otimes \rho_B$ and $\rho_B = \rho_0$. 

\begin{proof}[Proof of Proposition \ref{main:uniqueness of fixed points}]
    Suppose $\{\ket{i}_V\}_{0\le i \le n-1}$ is a standard basis of $\mc H_V$ and decompose $\rho_{VB}$ as
    \begin{equation}
        \rho_{VB} = \sum_{i,j} \ket{i}\bra{j}_V \otimes \rho_B^{ij}.
    \end{equation}
    For any $0\le k \le n-1$, using \eqref{fixed point}, we know that 
    \begin{equation}
    \begin{aligned}
       & \bra{k}_V  (id_{\mb B(\mc H_V)}\otimes \mc N(\rho_{VB})) \ket{k}_V  = \bra{k}_V \rho_{VB} \ket{k}_V = \mc N(\rho_B^{kk}) = \rho_B^{kk} \ge 0.
    \end{aligned}
    \end{equation}
    Therefore, if $\Tr(\rho_B^{kk}) \neq 0$, $\frac{\rho_B^{kk}}{\Tr(\rho_B^{kk}) }$ is a fixed point of $\mc N$ thus equal to $\rho_0$. If $\Tr(\rho_B^{kk}) = 0$, then we have $\rho_B^{kk} = 0 = 0\cdot \rho_0$. In summary, one has 
    \begin{equation}\label{diagonal}
        \rho_B^{kk} = \Tr(\rho_B^{kk}) \rho_0.
    \end{equation}
    For any $0\le k < l \le n-1$, define $\ket{\psi} = \alpha \ket{k}_V + \beta \ket{l}_V$, with $\alpha,\beta \in \mb C, |\alpha|^2 + |\beta|^2 = 1$. Using \eqref{fixed point}, we know that \begin{equation}
    \begin{aligned}
       & \bra{\psi}  (id_{\mb B(\mc H_V)}\otimes \mc N(\rho_{VB})) \ket{\psi}  = \bra{\psi} \rho_{VB} \ket{\psi}  = \mc N(|\alpha|^2 \rho_B^{kk}+ \bar{\alpha}\beta \rho_B^{kl} + \bar{\beta}\alpha \rho_B^{lk} + |\beta|^2 \rho_B^{ll}) \\
       & = |\alpha|^2 \rho_B^{kk}+ \bar{\alpha}\beta \rho_B^{kl} + \bar{\beta}\alpha \rho_B^{lk} + |\beta|^2 \rho_B^{ll} . 
    \end{aligned}
    \end{equation}
    Using the same argument as before, one has 
    \begin{align*}
        & |\alpha|^2 \rho_B^{kk}+ \bar{\alpha}\beta \rho_B^{kl} + \bar{\beta}\alpha \rho_B^{lk} + |\beta|^2 \rho_B^{ll} = \Tr\big(|\alpha|^2 \rho_B^{kk}+ \bar{\alpha}\beta \rho_B^{kl} + \bar{\beta}\alpha \rho_B^{lk} + |\beta|^2 \rho_B^{ll}\big) \rho_0.
    \end{align*}
    Recall the diagonal terms are proportional to $\rho_0$ \eqref{diagonal}, one has 
    \begin{align} \label{arbitrary}
        \bar{\alpha}\beta \rho_B^{kl} + \bar{\beta}\alpha \rho_B^{lk} =\Tr\big( \bar{\alpha}\beta \rho_B^{kl} + \bar{\beta}\alpha \rho_B^{lk} \big) \rho_0.
    \end{align}
    Since the choice of $\alpha,\beta$ in $\ket{\psi} = \alpha \ket{k}_V + \beta \ket{l}_V$ is arbitrary as long as $|\alpha|^2+|\beta|^2 = 1$, by $|\bar{\alpha} \beta| \le \frac{|\alpha|^2+|\beta|^2}{2}$, the range of $\bar{\alpha} \beta$ is given by \begin{equation}
        \{c \in \mb C: |c| \le \frac{1}{2}\}.
    \end{equation}
    Therefore, note that $\rho_{VB}$ is self-adjoint, we have $\rho_B^{kl} = (\rho_B^{lk})^{\dagger}$, \eqref{arbitrary} is equivalent to 
    \begin{equation}
        \forall c \in \mb C, |c|\le \frac{1}{2},\quad c \rho_B^{kl} + (c \rho_B^{kl})^{\dagger} = \Tr(c \rho_B^{kl} +  (c \rho_B^{kl})^{\dagger}) \rho_0,
    \end{equation}
    which implies \begin{align*}
        \rho_B^{kl} = \Tr(\rho_B^{kl}) \rho_0.
    \end{align*}
    In fact, denote $$\rho_B^{kl} = (x_{uv})_{0\le u,v \le \text{dim} B-1}, \rho_0 = (\rho_{uv})_{0\le u,v \le \text{dim} B-1}.$$ 
    Compare each element in the above equation, for any $0\le u,v \le \text{dim} B - 1$,
    \begin{align*}
        c(x_{uv} - \rho_{uv} \sum_r x_{rr}) + \overline{c} (\overline{x_{uv}} - \overline{\rho_{uv}\sum_r x_{rr}}) = 0.
    \end{align*}
    Since $|c| \le \frac{1}{2}$ can be any complex number, we must have 
    \begin{equation}
        x_{uv} = \rho_{uv} \sum_r x_{rr},
    \end{equation}
    which means $\rho_B^{kl} = \Tr(\rho_B^{kl}) \rho_0$. In summary, by showing that for any $0\le k,l \le n-1$, $\rho_B^{kl} = \Tr(\rho_B^{kl}) \rho_0$, we arrive at the conclusion
$ \rho_{VB}  = \sum_{i,j} \ket{i}\bra{j}_V \otimes \rho_B^{i,j} 
 = \sum_{i,j} \Tr(\rho_B^{ij})\ket{i}\bra{j}_V \otimes \rho_0 = \rho_V \otimes \rho_0.$
       
\end{proof}
The remaining task is to show the quantum channel $\mc N$ defined by \eqref{recovery channel} and \eqref{eqn:recove amplitude damping} has a unique fixed point. What we need is the following proposition, proved in \cite[Proposition 6.8]{wolf2012quantum}:
\begin{prop}\label{unique}
    Suppose $\mc N: \mb B(\mc H) \to \mb B(\mc H)$ is a (non-unital) quantum channel. If the Kraus representation of $\mc N$ given by 
    \begin{equation}
        \mc N(\rho): = \sum_{i\in I} E_i \rho E_i^{\dagger}
    \end{equation}
    satisfies: $\exists n \ge 1, \text{span}\{\prod_{k\le n}E_{i_k}:i_k \in I\} = \mb B(\mc H)$. Then $\mc N$ has a unique fixed point.
\end{prop}
Now it is straightforward to see that channel in \eqref{eqn:recove amplitude damping} has the Kraus representation 
\begin{equation}
    \mc N(\rho):= \sum_{i,j} A_{ij} \rho A_{ij}^{\dagger},
\end{equation}
where $A_{ij} = \rho_B^{1/2}E_i^{\dagger} \mc A_{\gamma'}(\rho_B)^{-1/2} E_j, i,j = 0,1$, $\gamma' = \frac{1-2\gamma}{1-\gamma}$ and $E_0 = \begin{pmatrix}
    1 & 0 \\
    0 & \sqrt{1-\gamma'}
\end{pmatrix},\ E_1 = \begin{pmatrix}
    0 & \sqrt{\gamma'} \\
    0 & 0
\end{pmatrix}$.

\noindent \textbf{Case 1:} If $\rho_B = \ketbra{0}$, we have $\mc A_{\gamma'}(\rho_B) = \ketbra{0}$. Therefore, the support of the recovery map is spanned by single vector $\ket{0}$ and it is trivial(identity). Therefore, in this case the equality condition is given by 
\begin{equation}
    \rho_{VB} = id_{\mb B(\mc H_V)}\otimes \mc A_{\gamma'}(\rho_{VB}).
\end{equation}
Note that $\mc A_{\gamma'}$ has a unique fixed point $\ketbra{0}$ thus $\rho_{VB} = \rho_V \otimes \ketbra{0}$.

\noindent \textbf{Case 2: }If $\rho_B = \begin{pmatrix}
    1-p & \delta \\
    \delta^* & p
\end{pmatrix}$ for $p\in (0,1)$, then denote 
\begin{align*}
    & \Delta_1 := p(1-p)-|\delta|^2, \\
    & \Delta_2 := (1-\gamma')\Delta_1 + \gamma'(1-\gamma')p^2, \\
    & \Delta := \frac{1}{\sqrt{\Delta_2(1+2\sqrt{\Delta_1})(1+2\sqrt{\Delta_2})}}.
\end{align*}
By direct calculation, the Kraus operators are given by 
\begin{align*}
    & \hspace{0.2cm}A_{00}/\Delta= (1-\gamma')\sqrt{\Delta_1}(\sqrt{\Delta_1} + p) + \sqrt{\Delta_2}(1-p + \sqrt{\Delta_1}) \ketbra{0}{0}  + \delta(1-\gamma')(\gamma'p + \sqrt{\Delta_2} - \sqrt{\Delta_1}) \ketbra{0}{1} \\
    & + (\sqrt{\Delta_2} - (1-\gamma')\sqrt{\Delta_1})\delta^*  \ketbra{1}{0} + (1-\gamma')(\Delta_1 + (p + \sqrt{\Delta_1})(\gamma'p + \sqrt{\Delta_2} + (1-p)\sqrt{\Delta_1})) \ketbra{1}{1}, \\
    & \hspace{0.2cm} A_{01} / \Delta =  \sqrt{\gamma'} \bigg(\delta^*(\sqrt{\Delta_2} - (1-\gamma')\sqrt{\Delta_1}) \ketbra{1}{1} + \left((1-\gamma')(\Delta_1 + p\sqrt{\Delta_1}) + \sqrt{\Delta_2}(\sqrt{\Delta_1} + 1-p)\right) \ketbra{0}{1} \bigg), \\
    & \hspace{0.2cm} A_{10} /\Delta  = \delta \sqrt{\gamma'}(\sqrt{\Delta_2} + (1-\gamma')p) \ketbra{0}{0}+ \sqrt{\gamma'}(1-\gamma')\delta^2 \ketbra{0}{1} \\
    & + \sqrt{\gamma'}(\sqrt{\Delta_1} + p)(\sqrt{\Delta_2} + (1-\gamma')p) \ketbra{1}{0}  -\delta \sqrt{\gamma'}(1-\gamma')(\sqrt{\Delta_1} + p) \ketbra{1}{1}, \\
    & \hspace{0.2cm} A_{11} /\Delta = \delta \gamma'(\sqrt{\Delta_2} + (1-\gamma')p) \ketbra{0}{1} + \gamma'(\sqrt{\Delta_2} + (1-\gamma')p)(p + \sqrt{\Delta_1})\ketbra{1}{1}.
\end{align*}
Using the fact that $\mc A_{\gamma'}(\rho_B)$ has full rank, one can directly check that 
\begin{equation}
    \text{span} \{A_{ij}\} = \mb M_2,
\end{equation}
thus using Proposition \ref{unique} we conclude the proof.

\section{Coherent information of generalized Platypus channels}\label{app:optimizer}
It is observed that the optimized state for $\mc Q^{(1)}(\mc N_{s,t})$ is diagonal with respect to the standard basis, i.e,
\begin{align*}
    & \quad \mc Q^{(1)}(\mc N_{s,t})= \max_{0\le u_0,u_1 \le u_0 + u_1 \le 1} I_c(\text{diag}(u_0,u_1,1-u_0-u_1), \mc N_{s,t}),
\end{align*}
which can be derived from the techniques in \cite{chessa2023resonant, leditzky2023platypus}. We can further improve the optimization as follows:
\begin{lemma}\label{lemma:restriction}
    $\mc Q^{(1)}(\mc N_{s,t})$ can be calculated as a single parameter optimization: 
    \begin{equation}\label{eqn:optimize coherent information}
    \begin{aligned}
       & \quad \mc Q^{(1)}(\mc N_{s,t}) = \begin{cases}
            \max_{0\le u \le 1} I_c(u \ketbra{0} + (1-u) \ketbra{1}, \mc N_{s,t}),\ s\ge \frac{1-t}{2}, \\
            \max_{0\le u \le 1} I_c(u \ketbra{0} + (1-u) \ketbra{2}, \mc N_{s,t}),\ s \le \frac{1-t}{2}.
        \end{cases}
    \end{aligned}
    \end{equation}
\end{lemma}
\begin{proof}
    The proof follows from a standard argument using majorization and Schur concavity of von Neumann entropy. For any fixed $u_0 \in [0,1]$, we claim that 
    \begin{align*}
        \max_{0\le u_1 \le 1- u_0} I_c(\text{diag}(u_0,u_1,1-u_0-u_1), \mc N_{s,t})
    \end{align*}
    is achieved either at $u_1= 0$ or $u_1 = 1-u_0$. In fact, denote $\rho_A = u_0 \ketbra{0} + u_1 \ketbra{1} + (1-u_0-u_1) \ketbra{2}$, using the formula in \eqref{eqn:matrix form of platypus}, we have 
    \begin{align*}
        & S(B) = S(\text{diag}(u_0 s, u_0(1-s-t),u_0 t + 1 - u_0)), \\
        & S(E) = S(\text{diag}(u_0 s + u_1, u_0(1-s-t)+ (1-u_0-u_1), u_0 t). 
    \end{align*}
    Note that $S(B)$ does not depend on $u_1$, thus we have 
    \begin{align*}
        & \max_{0\le u_1 \le 1- u_0} I_c(\text{diag}(u_0,u_1,1-u_0-u_1), \mc N_{s,t})\\
        & = S(B) - \min_{0\le u_1 \le 1- u_0} S(E).
    \end{align*}
    We claim that $\min_{0\le u_1 \le 1- u_0} S(E)$ is achieved at either $u_1 = 0$ or $u_1 = 1 - u_0$. Recall that for two Hermitian operators $H_1, H_2$ of the same size $d$, $H_1$ is majorized by $H_2$, denoted as $H_1 \prec H_2$, if 
    \begin{align*}
        & v^{\downarrow}(H_1) \prec v^{\downarrow}(H_2) \iff \sum_{j=1}^k v_j^1 \le \sum_{j=1}^k v_j^2, \quad \forall 1\le k \le d;\quad  \sum_{j=1}^d v_j^1 = \sum_{j=1}^d v_j^2,
    \end{align*}
    where $v^{\downarrow}(H_i)=  (v_1^i, v_2^i,\cdots, v_d^i)$ is the vector of singular values of $H_i$ with decreasing order: $v_1^i \ge v_2^i \ge\cdots \ge v_d^i$. By Schur concavity of von Neumann entropy, for any two density operators $\rho, \sigma$, $\rho \prec \sigma$ implies $S(\rho)\ge S(\sigma)$. Back to our claim, when $s \ge 1-s-t$, we can check that for any $0\le u_1 \le 1-u_0$, 
    \begin{equation}
    \begin{aligned}
    & \begin{pmatrix}
        u_0 s + u_1 & 0& 0 \\
        0 & u_0(1-s-t) + 1-u_0-u_1 & 0\\
        0 & 0& u_0 t
    \end{pmatrix}\prec \begin{pmatrix}
        u_0 s + 1-u_0 & 0& 0 \\
        0 & u_0(1-s-t) & 0\\
        0 & 0& u_0 t
    \end{pmatrix},
    \end{aligned}
    \end{equation}
    therefore by Schur concavity, we can show that $\min_{0\le u_1 \le 1- u_0} S(E)$ is achieved at $u_1 = 1- u_0$ in this case. Similarly, when $s \le 1-s-t$, we can show that $\min_{0\le u_1 \le 1- u_0} S(E)$ is achieved at $u_1 = 0$, which concludes the proof.
\end{proof}
\section*{Acknowledgment}
The authors thank the anonymous referees for their insightful comments and valuable suggestions, which helped improve the clarity and quality of this paper.

\bibliographystyle{marcotomPB}
\bibliography{id}

\end{document}